\pdfoutput=1
\documentclass[10pt,conference,final]{IEEEtran}
\usepackage{url}
\usepackage{booktabs}   %% For formal tables:
\usepackage{subcaption} %% For complex figures with subfigures/subcaptions
\usepackage{mathpartir}
\usepackage[all]{xy}
\usepackage{xspace}
\usepackage{amssymb}
\usepackage{amsmath}
\usepackage{amsthm}
\usepackage{stmaryrd}
\usepackage{xcolor}
\usepackage{comment}
\usepackage[sort,numbers]{natbib}
\usepackage[capitalise]{cleveref}
\usepackage{tikz-cd}
\usepackage{balance}

\newcommand{\R}{\mathbb{R}}
\newcommand{\N}{\mathbb{N}}
\newcommand{\B}{\mathbb{B}}

\newcommand{\prob}{\mathbb{P}}
\newcommand{\teq}{\triangleq}
\newcommand{\with}{\mathbin{\&}}

\newcommand{\Met}{\mathsf{Met}}

\newcommand{\Set}{\mathsf{Set}}
\newcommand{\RSRel}{\mathsf{RSRel}}

\newcommand{\db}{\mathsf{db}}
\newcommand{\id}{\mathrm{id}}

\newcommand{\dist}{D}
\newcommand{\distmonad}{\mathcal{D}}
\newcommand{\T}{T}
\newcommand{\Tcons}{\bigcirc}

\renewcommand{\nearrow}{\mathbin{\stackrel{\rm ne}{\arrow}}}
\newcommand{\rearrow}{\mathbin{\stackrel{\rm re}{\arrow}}}
\newcommand{\darrow}{\mathbin{\dot{\arrow}}}
\newcommand{\dlarrow}{\mathbin{\dot{\multimap}}}

\def\blet{\mathop{\textbf{let}}}

\def\bbind{\mathop{\textbf{bind}}}

\def\bin{\mathop{\textbf{in}}}

\def\btrue{\mathop{\textbf{true}}}
\def\textbfalse{\mathop{\textbf{false}}}

\def\binl{\mathop{\textbf{inl}}\nolimits}
\def\binr{\mathop{\textbf{inr}}\nolimits}

\def\bcase{\mathop{\textbf{case}}\nolimits}
\def\bof{\mathop{\textbf{of}}}

\def\myrel#1\over#2{\mathrel{\mathop{\kern0pt#2}\limits_{#1}}}

\def\breturn{\mathop{\textbf{return}}}

\def\rset{\mathop{\textsf{set}}}

\def\kl{\mathrm{kl}}

\def\lapD{\mathop{\textbf{Laplace}}}
\def\berD{\mathop{\textbf{Bernoulli}}}
\def\poiD{\mathop{\textbf{Poisson}}}
\def\gauD{\mathop{\textbf{Gaussian}}}
\def\norD{\mathop{\textbf{Normal}}}

\def\dMax{\mathsf{MD}}
\def\dStat{\mathsf{SD}}
\def\dSkew{\mathsf{AD}}
\def\dKL{\mathsf{KL}}
\def\dHellinger{\mathsf{HD}}
\def\dChi{\mathsf{XD}}

\def\rDP{\mathsf{DPR}}
\def\rKL{\mathsf{KLR}}
\def\rHD{\mathsf{HDR}}
\def\rChi{\mathsf{XDR}}

\newcommand{\mT}{\mathcal{T}}
\newcommand{\arrow}{\rightarrow}
\newcommand{\Arrow}{\Rightarrow}
\newcommand{\klift}[1]{{#1}^\dagger}
\newcommand{\pklift}[1]{{#1}^\ddagger}

\newcommand{\unit}{{\bf I}}
\newcommand{\ox}{\mathbin{\otimes}}
\newcommand{\dunit}{\dot{\bf  I}}
\newcommand{\dox}{\mathbin{\dot\otimes}}
\newcommand{\dmultimap}{\mathbin{\dot\multimap}}
\newcommand{\dtimes}{\mathbin{\dot\times}}

\newcommand{\dpitchfork}{\mathbin{\dot\pitchfork}}

\newcommand{\ddunit}{\ddot{\bf  I}}
\newcommand{\ddox}{\mathbin{\ddot\otimes}}
\newcommand{\ddpitchfork}{\mathbin{\ddot\pitchfork}}

\newcommand{\fa}[1]{\forall{#1}~.~}

\newcommand{\sem}[1]{\llbracket #1\rrbracket}

\newcommand{\adjunction}[3]{
  \ar@<.4pc>[#1]^-{#2} \ar@{}[#1]|-*=0[@]{\bot} \ar@<-.4pc>@{<-}[#1]_-{#3}
}
\newcommand{\EE}{\mathbb{E}}
\newcommand{\FF}{\mathbb{F}}
\newcommand{\BB}{\mathbb{B}}
\newcommand{\CC}{\mathbb{C}}
\newcommand{\DD}{\mathbb{D}}

\newcommand{\BLift}[1]{\mathbf{!Lift}_{#1}}
\newcommand{\RLift}[1]{\mathbf{RLift}_{#1}}
\newcommand{\Asign}[1]{\mathbf{Asign}_{#1}}
\newcommand{\Comp}{\mathbf{Comp}}
\newcommand{\Lift}[1]{\mathbf{Lift}_{#1}}
\newcommand{\Ord}{\mathbf{Ord}}

\newcommand{\CLat}{\mathbf{CLat}}

\renewcommand{\epsilon}{\varepsilon}

\newcommand{\rname}[1]{\quad #1}
\newcommand{\brname}[1]{(#1)\xspace}

\newcommand{\ruleconst}{\ensuremath{\brname{\mathrm{Const}}}}

\newcommand{\rulevar}{\ensuremath{\brname{\mathrm{Var}}}}

\newcommand{\ruleitens}{\ensuremath{\brname{\otimes I}}}
\newcommand{\ruleetens}{\ensuremath{\brname{\otimes E}}}
\newcommand{\ruleiamp}{\ensuremath{\brname{\& I}}}
\newcommand{\ruleeamp}{\ensuremath{\brname{\& E}}}
\newcommand{\ruleisuml}{\ensuremath{\brname{+ I_l}}}
\newcommand{\ruleisumr}{\ensuremath{\brname{+ I_r}}}
\newcommand{\ruleesum}{\ensuremath{\brname{+ E}}}
\newcommand{\ruleiapp}{\ensuremath{\brname{\multimap I}}}
\newcommand{\ruleeapp}{\ensuremath{\brname{\multimap E}}}
\newcommand{\ruleibang}{\ensuremath{\brname{! I}}}
\newcommand{\ruleebang}{\ensuremath{\brname{! E}}}
\newcommand{\ruleiunit}{\ensuremath{\brname{1I}}}

\newtheorem{theorem}{Theorem}
\newtheorem{proposition}{Proposition}
\newtheorem{corollary}{Corollary}
\newtheorem{lemma}{Lemma}
\newtheorem{definition}{Definition}

\newcommand{\dfuzzvar}{
  \inferrule
  { (x :_r \sigma) \in \Gamma \\ r \geq 1 }
  { \Gamma \vdash x : \sigma }
  \rname{\rulevar}
}
\newcommand{\dfuzzconst}{
  \inferrule
  { k \in \R }
  { \Gamma \vdash k : \R }                    \rname{\ruleconst}
}
\newcommand{\dfuzziunit}{
  \inferrule
  { }
  { \Gamma \vdash () : 1 } \rname{\ruleiunit}
}
\newcommand{\dfuzzitens}{
  \inferrule
  { \Gamma \vdash e_1 : \sigma \\ \Delta \vdash e_2 : \tau }
  { \Gamma + \Delta \vdash (e_1, e_2) : \sigma \otimes \tau } \rname{\ruleitens}
}
\newcommand{\dfuzzetens}{
  \inferrule
  { \Gamma \vdash e : \sigma_1 \otimes \sigma_2  \\
    \Delta, x :_r \sigma_1, y :_r \sigma_2 \vdash e' : \tau }
  { r\Gamma + \Delta \vdash \blet {(x, y)} = e \bin e' : \tau} \rname{\ruleetens}
}
\newcommand{\dfuzziamp}{ \inferrule { \Gamma \vdash e_1 : \sigma \\ \Gamma \vdash e_2 : \tau }
  { \Gamma \vdash \langle e_1, e_2 \rangle : \sigma \with \tau } \rname{\ruleiamp}
}
\newcommand{\dfuzzeamp}{
  \inferrule
  { \Gamma \vdash e : \sigma_1 \with \sigma_2 }
  { \Gamma \vdash \pi_i \;e : \sigma_i }                  \rname{\ruleeamp}
}
\newcommand{\dfuzziapp}{
  \inferrule
    { \Gamma, x :_1 \sigma \vdash  e : \tau }
  { \Gamma \vdash  \lambda x.\, e : \sigma \multimap \tau } \rname{\ruleiapp}
}
\newcommand{\dfuzzeapp}{
  \inferrule
  { \Gamma \vdash  e_1 : \sigma \multimap \tau \\
    \Delta \vdash  e_2 : \sigma }
  { \Gamma + \Delta \vdash e_1\;e_2 : \tau }           \rname{\ruleeapp}
}
\newcommand{\dfuzzibang}{
  \inferrule
  { \Gamma \vdash e : \sigma }
  { r\Gamma \vdash {{!}e} : {{!_r}\sigma} } \rname{\ruleibang}
}
\newcommand{\dfuzzebang}{
  \inferrule
  { \Gamma \vdash e_1 : {{!_s}\sigma} \\ \Delta, x :_{rs} \sigma \vdash e_2 : \tau}
  { r\Gamma + \Delta \vdash \blet {{!}x} = e_1 \bin e_2 : \tau } \rname{\ruleebang}
}
\newcommand{\dfuzzisuml}{
  \inferrule
  { \Gamma \vdash e : \sigma }
  { \Gamma \vdash \binl e : \sigma + \tau } \rname{\ruleisuml}
}
\newcommand{\dfuzzisumr}{
  \inferrule
  { \Gamma \vdash e : \tau }
  { \Gamma \vdash \binr e : \sigma + \tau } \rname{\ruleisumr}
}
\newcommand{\dfuzzesum}{
  \inferrule
  { \Gamma \vdash e : \sigma_1 + \sigma_2 \\
    \Delta, x :_r \sigma_1 \vdash e_l : \tau \\
    \Delta, y :_r \sigma_2 \vdash e_r : \tau }
  { r\Gamma + \Delta \vdash \bcase e \bof
    \binl x.\, e_l
    \mid \binr y.\, e_r : \tau } \rname{\ruleesum}
}

\usepackage[margin=false,inline=true]{fixme}
\FXRegisterAuthor{aa}{anaa}{\color{teal}AA}
\FXRegisterAuthor{gb}{angb}{\color{blue}GB}
\FXRegisterAuthor{mg}{anmg}{\color{magenta}MG}
\FXRegisterAuthor{jh}{anjh}{\color{red}JH}
\FXRegisterAuthor{sk}{ansk}{\color{blue}SK}

\def\aa{\aanote}

\def\mg{\mgnote}

\def\sk{\sknote}

\begin{document}

% \jh{Alternatives:}
\title{Probabilistic Relational Reasoning via Metrics}
% \title{Probabilistic Relational and Metric Reasoning}
% \title{Combining Relational and Metric Reasoning \\ for
% Probabilistic Programs}

\author{
  \IEEEauthorblockN{Arthur Azevedo de Amorim}
  \IEEEauthorblockA{
    Carnegie Mellon University\\
    Pittsburgh, PA
  } \and
  \IEEEauthorblockN{Marco Gaboardi}
  \IEEEauthorblockA{
    University at Buffalo\\
    Buffalo, NY
  } \and
  \IEEEauthorblockN{Justin Hsu}
  \IEEEauthorblockA{
    University of Wisconsin\\
    Madison, WI
  } \and
  \IEEEauthorblockN{Shin-ya Katsumata}
  \IEEEauthorblockA{
    National Institute of Informatics\\
    Tokyo, Japan
  }
}

\IEEEoverridecommandlockouts
\IEEEpubid{\makebox[\columnwidth]{978-1-7281-3608-0/19/\$31.00~
    \copyright2019 IEEE \hfill} \hspace{\columnsep}\makebox[\columnwidth]{ }}
\maketitle

\begin{abstract}
  The Fuzz programming language by Reed and Pierce uses an elegant linear type
  system combined with a monad-like type to express and reason about
  probabilistic sensitivity properties, most notably $\epsilon$-differential
  privacy. We show how to extend Fuzz to capture more general \emph{relational}
  properties of probabilistic programs, with \emph{approximate}, or
  $(\epsilon, \delta)$-differential privacy serving as a leading example. Our
  technical contributions are threefold. First, we introduce the categorical
  notion of \emph{comonadic lifting} of a monad to model composition properties
  of probabilistic divergences. Then, we show how to express relational
  properties in terms of sensitivity properties via an adjunction we call the
  \emph{path construction}.  Finally, we instantiate our semantics to model the
  terminating fragment of Fuzz extended with types carrying information about
  other divergences between distributions.
\end{abstract}

\section{Introduction}
\label{sec:introduction}

Over the past decade, differential privacy has emerged as a robust,
compositional notion of privacy, imposing rigorous bounds on what
database queries reveal about private data. Formally, a probabilistic
database query $f$ is $\epsilon$-differentially private if, given two
pairs of \emph{adjacent} databases $X_1$ and $X_2$---that is,
databases differing in at most one record---we have
\begin{align}
  \prob(f(X_i) \in U) \leq e^\epsilon\prob(f(X_j) \in U) \quad (i,j \in \{1,2\}),
  \label{eq:epsilon-dp}
\end{align}
where $U$ ranges over arbitrary sets of query results. Intuitively,
the parameter $\epsilon$ measures how different the result
distributions $f(X_1)$ and $f(X_2)$ are---the smaller $\epsilon$ is,
the less the output depends on any single record in the input
database.

% \aa{Could we list some properties here?}  \mg{I am not sure which
% properties we are talking about here. this is not really my
% understanding of the contributions of this paper.}  \jh{To be as
% precise as possible, properties of the form ``related inputs lead to
% two output distributions with bounded divergence''. These properties
% include $(\epsilon, \delta)$-privacy of course, but also for other
% divergences.}

The strengths and applications of differential privacy have prompted the
development of a range of verification techniques, in particular approaches
based on linear types such as the Fuzz language~\citep{Reed:2010}.  These
systems exploit the fact that $\epsilon$-differential privacy is equivalent to a
\emph{sensitivity} property, a guarantee that applies to \emph{arbitrary} pairs
of databases:
\begin{equation} \label{eq:epsilon-dp-sensitivity} \dMax(f(X_1),
  f(X_2)) \leq \epsilon \cdot d_{DB}(X_1, X_2),
\end{equation}
where $d_{DB}$ measures how similar the input databases are (for instance, via
the Hamming distance on sets), and $\dMax$, the \emph{max divergence}, is the
smallest value of $\epsilon$ for which \eqref{eq:epsilon-dp} holds. In other
words, $f$ is $\epsilon$-Lipschitz continuous, or
\emph{$\epsilon$-sensitive}. Sensitivity has pleasant properties for formal
verification; for example, the sensitivity of the composition of two functions
is the product of their sensitivities. By leveraging these composition
principles, Fuzz can track the sensitivity of a function in its type, reducing a
proof of differential privacy for an algorithm to simpler sensitivity checks
about its components.

However, not all distance bounds between distributions can be converted into
sensitivity properties. One example is $(\epsilon, \delta)$-differential privacy
\citep{dwork2006our}, a relaxation of $\epsilon$-differential privacy that allows
privacy violations with a small probability $\delta$; in return,
$(\epsilon, \delta)$-differential privacy can allow significantly more accurate
data analyses. Superficially, its definition resembles \eqref{eq:epsilon-dp}: a
query $f$ is $(\epsilon, \delta)$-differentially private if, for all pairs of
adjacent input databases $X_1$ and $X_2$, we have
\begin{align}
  \prob(f(X_i) \in U) \leq e^\epsilon\prob(f(X_j) \in U) + \delta \quad
  (i, j \in \{1,2\}).
  \label{eq:epsilon-delta-dp}
\end{align}
Setting $\delta = 0$ recovers the original definition. Introducing the
\emph{skew divergence} $\dSkew_\epsilon$ \citep{BartheO13},
$(\epsilon, \delta)$-privacy is equivalent to a bound
$\dSkew_\epsilon(f(X_1), f(X_2)) \leq \delta$ for adjacent databases.

Despite the similarity between the two definitions, Fuzz could not handle
$(\epsilon, \delta)$-differential privacy, because it cannot be stated directly
in terms of function sensitivity.  This is possible for $\epsilon$-differential
privacy because \eqref{eq:epsilon-dp} can be recast as the bound
$\dMax(f(X_1), f(X_2)) \leq \epsilon$, which is equivalent to the sensitivity
property \eqref{eq:epsilon-dp-sensitivity} because the max divergence is a
proper metric satisfying the triangle inequality.  In contrast, the skew
divergence does \emph{not} satisfy the triangle inequality and does not scale up
smoothly when the inputs $X_1$ and $X_2$ are farther apart---for instance, an
$(\epsilon, \delta)$-private function $f$ usually does not satisfy
$\dSkew_{\epsilon}(f(X_1), f(X_2)) \leq 2 \cdot \delta$ when $X_1$ and $X_2$ are
at distance $2$.  Similar problems arise for other properties based on distances
that violate the triangle inequality, such as the Kullback-Leibler (KL) and
$\chi^2$ divergences.

This paper aims to bridge this gap, showing that Fuzz's core can \emph{already}
accommodate other quantitative properties, with $(\epsilon, \delta)$-privacy
being our motivating application. To do so, we first need a semantics for the
probabilistic features of Fuzz. Typically in programming language semantics,
probabilistic programs are structured using the probability
monad~\citep{Moggi89,Giry1982}: the return operation produces a deterministic
distribution that always yields the same value, while the bind operation samples
from a distribution and runs another probabilistic computation. However, monads
cannot describe the composition principles supported by many useful metrics on
probabilities---though the typing rules for distributions in Fuzz resemble the
usual monadic rules~\citep{Moggi89}, there issues related to the context
sensitivities. Accordingly, our \emph{first contribution} is a notion of
\emph{comonadic lifting} of a monad, which lifts the operations of a monad from
a symmetric monoidal closed category (SMCC) to a related \emph{refined}
category. We demonstrate our theory by modeling statistical distance and Fuzz's
max divergence. We also propose a \emph{graded} variant of liftings to encompass
other examples, including the Hellinger distance and the KL and $\chi^2$
divergences.

Our \emph{second contribution} is a \emph{path metric} construction
that reduces relational properties such as
$(\epsilon, \delta)$-differential privacy to equivalent statements
about sensitivity. Concretely, given any reflexive, symmetric relation
$R$ on a set $X$, we define the path metric $d_R$ on $X$ by setting
$d_R(x_1, x_2)$ to be the length of the shortest path connecting $x_1$
and $x_2$ in the graph corresponding to $R$. The path construction
provides a full and faithful functor from the category $\RSRel$ of
reflexive, symmetric relations into the category $\Met$ of metric
spaces and 1-sensitive functions. We also show a right adjoint to the
path construction---as $\Met$ is a symmetric monoidal closed category
and $\RSRel$ is a cartesian closed category, this adjunction recalls
mixed linear and non-linear models of linear logic
\citep{benton1994mixed}.

% \aa{This would be more convincing if we gave examples of other
% relational properties that we could care about.}

Putting these two pieces together, our \emph{third contribution} is a model of
the terminating fragment of Fuzz \citep{Reed:2010}. We extend the language with
new types and typing rules to express and reason about relational properties
beyond $\epsilon$-differential privacy, including
$(\epsilon, \delta)$-differential privacy. Our framework can smoothly
incorporate the new features by combining the path construction and graded
liftings, giving a unified perspective on a class of probabilistic relational
properties.

\paragraph*{Outline}
We begin by reviewing the Fuzz language, the interpretation of its deterministic
fragment in the category of metric spaces, and some basic probability theory in
\cref{sec:preliminaries}. \cref{sec:liftings} introduces comonadic liftings on
monads, and uses them to interpret the terminating fragment of Fuzz with
probabilistic constructs.

Shifting gears, \cref{sec:rel} explores how probabilistic properties can be
modeled in the category of relations. \cref{sec:grading} develops a graded
version of our comonadic liftings over relations, to model composition of
properties like $(\epsilon, \delta)$-privacy. In \cref{sec:path} we consider how
to transfer liftings between different categories; our leading example is the
\emph{path construction}, which moves liftings over relations to liftings over
metric spaces. As an application, \cref{sec:grfuzz} extends Fuzz with
graded types capable of modeling $(\epsilon, \delta)$-privacy and other
sensitivity properties of divergences.

In \cref{sec:limitations}, we sketch how our results can be partially extended
to model general recursion in Fuzz, by combining metric CPOs \citep{AGHKC17}
with the probabilistic powerdomain of \citet{JP89}. While the probabilistic
features and liftings pose no problems, the path construction runs into
technical difficulties; we leave this extension as a challenging open
problem. Finally, we survey related work (\cref{sec:rw}) and conclude
(\cref{sec:conclusions}).

\section{Preliminaries}
\label{sec:preliminaries}

To fix notation and terminology, we review here basic concepts of
category theory, probabilities and metric spaces.  For ease of
reference, we include an overview of the deterministic, terminating
fragment of the Fuzz language, recalling its semantics based on metric
spaces~\citep{AGHKC17}.

\subsection{Category Theory}

% Given two preordered sets (or classes) $X$ and $Y$, by $[X,Y]$ we mean
% the preordered set (or class) of monotone mappings from $X$ to $Y$.

A map of adjunctions \citep[Sections IV.1,IV.7]{cwm2} from
$\langle L,R,\eta,\epsilon\rangle:\CC\rightharpoonup\DD$ to
$\langle L',R',\eta',\epsilon'\rangle:\CC'\rightharpoonup\DD'$
% \begin{displaymath}
%   \left(\xymatrix{ \CC & \DD \adjunction{l}{L}{R} },\eta,\epsilon\right),\quad
%   \left(\xymatrix{ \CC' & \DD' \adjunction{l}{L'}{R'} },\eta',\epsilon'\right),
% \end{displaymath}
is a pair of functors $F:\CC\arrow\CC',G:\DD\arrow\DD'$ satisfying
$ G\circ L=L'\circ F, F\circ R=R'\circ G$ and $ F\circ\eta=\eta'\circ F.  $ (The
last equality is equivalent to $G\circ\epsilon=\epsilon'\circ G$.) We write such
a map as $(F,G):(L\dashv R)\arrow (L'\dashv R')$.

Though our main applications revolve around probabilistic programs, we develop
our theory in terms of general monads $\mT = (\T,\eta,\klift{(-)})$ presented as
Kleisli triples.  The operation $\klift{(-)}$, the \emph{Kleisli lifting} of the
monad, promotes a morphism $f : X \arrow \T Y$ to $f^\dagger : \T X \arrow \T Y$
and satisfies common unit and associativity laws.  We assume that $\mT$ is
defined on a symmetric monoidal closed category and carries a compatible
\emph{strength} $\sigma_{X,Y} : X \ox \T Y \arrow \T (X \ox Y)$
\cite{kockstrong}, which allows us to define \emph{parameterized} liftings:
given $f : X \ox Y \arrow \T Z$, we define $\pklift{f} : X \ox \T Y \arrow \T Z$
as the composite
\begin{align}
  \label{eq:parameterized-kleisli}
  \pklift{f} \teq \xymatrix{ X \ox \T Y \ar[r]^\sigma & \T (X \ox Y) \ar[r]^{\klift{f}} & \T Z}.
\end{align}
Lifting the evaluation morphism
$ev_{X,\T Y} : (X \multimap \T Y) \ox X \arrow \T Y$ of the internal
hom $\multimap$ yields the internalized Kleisli lifting
\begin{equation}
  \label{eq:internal-kleisli}
  \kl_{X,Y}^{\mT} : (X \multimap \T Y) \ox \T X \arrow \T Y.
\end{equation}

\subsection{Probability Theory}
\label{sec:probabilities}

Our running example is the monad $\dist X$ of discrete probability
distributions over a set $X$---that is, functions $\mu : X \to [0,1]$
such that $\mu(x) \neq 0$ for at most countably many elements
$x \in X$, and $\sum_{x \in X} \mu(x) = 1$.  For a subset
$U \subseteq X$, we define $\mu(U)$ as $\sum_{x \in U} \mu(x)$.  Given
an element $x \in X$, we write $\eta(x) \in \dist X$ for the point
mass at $x$, i.e., $\eta(x)(x') \teq 1$ if $x = x'$, otherwise 0.
Given $f : X \to \dist Y$, its Kleisli lifting
$f^\dagger : \dist X \to \dist Y$ is defined by sampling a value from
its input distribution and feeding that sample to $f$.  Formally,
\begin{equation*}
  f^\dagger(\mu)(y) \teq \sum_{x\in X}f(x)(y)\mu(x)\quad
  (\mu\in\dist X,y\in Y).
\end{equation*}

\subsection{Metric Spaces}
Let $\R^\infty_{\geq 0}$ be the set of non-negative reals extended
with a greatest element $\infty$.  A \emph{metric} on a set $X$ is a
function $d : X \times X \to \R^{\infty}_{\geq 0}$ that satisfies
(i) $d(x, x) = 0$, (ii) $d(x, y) = d_X(y, x)$, and (iii)
$d(x, z) \leq d(x, y) + d(y, z)$ (the \emph{triangle
    inequality}).
We stipulate that $r \bullet \infty = \infty \bullet r = \infty$ for
any $r \in \R^\infty_{\geq 0}$, where $\bullet$ stands for addition or
multiplication.\mg{I think we should comment on the difference from
  POPL18 since there we claimed that solution was the ``good''
  one. Maybe just anticipate the note from the next page?}
A \emph{metric space} is a pair $X = (|X|, d_X)$, where $|X|$ is a
carrier set and $d_X$ is a metric on $|X|$.\footnote{%
  These conditions technically define an \emph{extended pseudo-metric
    space}---``extended'' because distances may be infinite, and
  ``pseudo'' because distinct elements may be at distance 0---but for
  the sake of brevity we use ``metric space'' throughout.}  We will
often refer to a metric space by its carrier, and we will write $d$
with no subscript when the metric space is clear. A metric
space $X$ is {\em above} a set $I$ if $|X|=I$, that is, $X=(I,d)$ with
some metric $d$ on $I$.

A function $f : X \to Y$ between two metric spaces is
\emph{$r$-sensitive} if $d_Y(f(x_1), f(x_2)) \leq rd_X(x_1, x_2)$ for
all pairs of elements $x_1, x_2 \in X$.  Thus, the smaller $r$ is, the
less the output of a function varies when its input varies.  Note that
this condition is vacuous when $r = \infty$, so any function between
metric spaces is $\infty$-sensitive. When $r = 1$, we speak of a
\emph{non-expansive} function instead.  We write $f:X\nearrow Y$ to
mean that $f$ is a non-expansive function from $X$ to $Y$.

To illustrate these concepts, consider the set of real numbers $\R$
equipped with the Euclidean metric: $d(x, y) = |x - y|$.  The doubling
function that maps the real number $x$ to $2x$ is 2-sensitive; more
generally, a function that scales a real number by another real number
$k$ is $|k|$-sensitive.  However, the squaring function that maps each
$x$ to $x^2$ is not $r$-sensitive for any finite $r$. The identity
function on a metric space is always non-expansive.  The definition of
$\epsilon$-differential privacy, as stated in
\eqref{eq:epsilon-dp-sensitivity}, says that the private query $f$ is
$\epsilon$-sensitive.

\begin{figure}
  \centering
  \begin{tabular}{ccc} \toprule Metric Space & Carrier Set &
    {$d(a, b)$} \\ \midrule
    {$\R$}             & {$\R$}        & {$|a - b|$} \\
    {$1$}              & {$\{\star\}$} & {$0$} \\
    {$r \cdot X$}      & {$X$}         &
                                         {$\begin{cases}
                                             rd_X(a, b)  & : r \neq \infty \\
                                             \infty      & : r = \infty, a \neq b \\
                                             0           & : r = \infty, a =  b
                                           \end{cases}$} \\
    {$X \times Y$}   & {$X \times Y$} & {$\max(d_X(a_1, b_1), d_Y(a_2, b_2))$} \\
    {$X \otimes Y$} & {$X \times Y$} & {$d_X(a_1, b_1) + d_Y(a_2, b_2)$} \\
    %
    % {$\prod_{i \in I} X_i$} & {$\prod_{i \in I} X_i$} &
    % {$\sup_{i\in I}d_{X_i}(a(i),b(i))$} \\[0.5em]
    % {$I\pitchfork X$} & {fns. $I\arrow Y$} &
    % {$\sup_{i\in I} d_Y(a(i),b(i))$} \\
    %
    {$X + Y$} & {$X \uplus Y$} &
                                 {$\begin{cases}
                                     d_X(a, b) & \text{ if $a, b \in X$} \\
                                     d_Y(a, b) & \text{ if $a, b \in Y$} \\
                                     \infty & \text{ otherwise}
                                   \end{cases}$} \\[0.5em]
    {$X \multimap Y$}
                                             & $X \to Y$
                                                           & {$\sup_{x \in X}d_Y(a(x),b(x))$} \\
    & {Non-exp.} & \\
    \bottomrule
  \end{tabular}
  \caption{Basic constructions on metric spaces}
  \label{fig:metric-spaces}
\end{figure}

\Cref{fig:metric-spaces} summarizes basic constructions on metric
spaces.  Scaling allows us to express sensitivity in terms of
non-expansiveness: an $r$-sensitive function from $X$ to $Y$ is a
non-expansive function from the scaled metric space $r \cdot X$ (or
simply $rX$) to $Y$. The metric on $\infty X$ does not depend on the
metric of $X$---note that we define scaling by $\infty$ separately
from scaling by a finite number to ensure that $d(x,x) = 0$---so we
use this notation even when $X$ is a plain set that does not have a
metric associated with it.  We consider two metrics on products: one
combines metrics by taking their maximum, while the other takes their
sum. The metric on disjoint unions places their two sides infinitely
apart.

Let $\Met$ be the category of metric spaces and non-expansive functions. We
write $p:\Met\arrow\Set$ to denote the forgetful functor defined by $pX=|X|$ and
$pf=f$. Note that $\Set(X, pY) = \Met(\infty\cdot X, Y)$, so $\infty \cdot (-)$
and $p$ form an adjoint pair. The $\times$ metric---used to interpret the
connective $\with$---and $+$ metric yield products and sums in $\Met$, with the
expected projections and injections. Since the two sides of a sum are infinitely
apart, we can define non-expansive functions by case analysis without reasoning
about sensitivity across two different branches. The other metric on products,
given by $\otimes$, is needed to make the operations of currying and function
application compatible with the metric on non-expansive functions defined
above. Formally, $\Met$ forms a symmetric monoidal closed category with monoidal
structure given by $(\otimes, 1)$ and exponentials given by $\multimap$.

\subsection{The Fuzz Language}
Fuzz \citep{Reed:2010} is a type system for analyzing program sensitivity. The
language is a largely standard, call-by-value lambda calculus;
\cref{fig:fuzz-syntax} summarizes the syntax, types, contexts, and context
operations. Types in Fuzz are interpreted as metric spaces, and function types
carry a numeric annotation that describes their sensitivity. The type system
tracks the sensitivity of typed terms with respect to each of its bound
variables, akin to bounded linear logic~\citep{GSS92}. \cref{fig:fuzz-typing}
presents the typing rules of the terminating, deterministic fragment of the
language.

In prior work \citep{AGHKC17}, we described a model where each typing derivation
$x_1 :_{r_1} \tau_1,\ldots,x_n :_{r_n} \tau_n \vdash e : \sigma$ corresponds to
a non-expansive function
$\sem{e} : r_1\sem{\tau_1} \ox \cdots \ox r_n\sem{\tau_n} \to
\sem{\sigma}$,\footnote{%
  Our interpretation of scaling differs slightly from our previous
  one~\citep{AGHKC17} in that distinct points in scaled spaces $\infty X$ are
  infinitely apart.  This does not affect the validity of the interpretation; in
  particular, scaling is still associative, and commutes with $\ox$, $\with$ and
  $+$.} where types are interpreted homomorphically using the constructions on
metric spaces described thus far. In particular, context splitting corresponds
to the family of functions
$\delta : \sem{\Delta+\Gamma}\nearrow\sem\Delta\ox\sem\Gamma$ defined as
$\delta(x) \teq (x_1, x_2)$, where $x_1$ and $x_2$ are obtained by removing the
components of $x$ that do not appear in $\Delta$ and $\Gamma$, respectively.

In addition to the probabilistic features that we will cover next, the original
Fuzz language also includes general recursive types.  These
pose challenges related to non-termination, which we return to in
\cref{sec:limitations}.

\begin{figure}
  \centering
  \begin{align*}
    e \in E & ::= x
              \mid k \in \R
              \mid e_1 + e_2
              \mid ()
              \mid \lambda x.\, e
              \mid e_1\; e_2
              \mid (e_1, e_2)
    \\
            & \mid \blet {(x, y)} = e \bin e'
              \mid \langle e_1, e_2 \rangle
              \mid \pi_i\; e
              \mid {{!}e}
              \mid \blet {{!}x} = e \bin e'
    \\
            & \mid \binl e
              \mid \binr e
              \mid (\bcase e \bof \binl x.\, e_l \mid \binr y.\, e_r)
              % \mid \bfold e
              % \mid \bunfold e
    \\
    % v \in V & ::= k \in \R
    %           \mid ()
    %           \mid \lambda x.\, e
    %           \mid (v_1, v_2)
    %           \mid \langle v_1, v_2 \rangle
    %           \mid {!v}
    %           \mid \binl v
    %           \mid \binr v
    %           % \mid \bfold v
    % \\
    \sigma, \tau & ::=
    % \alpha \in T \mid
                   \R
                   \mid 1
                   \mid \sigma \multimap \tau
                   \mid \sigma \otimes \tau
                   \mid \sigma \with \tau
                   \mid \sigma + \tau
                   \mid {{!_r}\sigma}
  \end{align*}
  \hrule
  \[
    r, s \in \R^\infty_{\geq 0} \qquad
    \Gamma, \Delta ::= \varnothing \mid \Gamma, x :_r \sigma
    % AAA: No recursive types here.
    % \Phi ::= (\alpha \mapsto \Phi(\alpha))_{\alpha \in T} \\
  \]
  \[
      r \cdot \varnothing \teq \varnothing \qquad
      r \cdot (\Gamma, x :_s \sigma) \teq r \cdot \Gamma, x :_{r \cdot s} \sigma
    \]
  \begin{gather*}
      \varnothing + \varnothing \teq \varnothing \qquad
      (\Gamma, x :_r \sigma) + (\Delta, x :_s \sigma) \teq (\Gamma + \Delta), x :_{r + s} \sigma \\
      (\Gamma, x :_r \sigma) + \Delta \teq (\Gamma + \Delta), x :_r \sigma \qquad (x \notin \Delta) \\
      \Gamma + (\Delta, x :_s \sigma) \teq (\Gamma + \Delta), x :_s \sigma \qquad (x \notin \Gamma)
  \end{gather*}
  \caption{Fuzz syntax, types, contexts, and context operations}
  \label{fig:fuzz-syntax}
\end{figure}

\begin{figure}
  \begin{mathpar}
    \dfuzzvar \and \dfuzziunit \and \dfuzzconst \and
    % \dfuzzplus \\
    \dfuzziapp \and \dfuzzeapp \and \dfuzzitens \and \dfuzzetens \and
    \dfuzziamp \and \dfuzzeamp \and \dfuzzibang \and \dfuzzebang \and
    \dfuzzisuml \and \dfuzzisumr \and \dfuzzesum % \\
    % AAA: No recursive types here.
    % \dfuzzimu \and
    % \dfuzzemu \\
  \end{mathpar}
  \caption{Fuzz typing rules (deterministic, terminating fragment)}
  \label{fig:fuzz-typing}
\end{figure}

%\section{Extending Semantics to Handle Probabilities}
\section{A Semantics for Probabilistic Fuzz}
\label{sec:liftings}

In addition to the deterministic constructs, Fuzz offers a monad-like interface
for probabilistic programming~\citep{Moggi89}, structured around the operations
discussed in \cref{sec:probabilities}---sampling and producing a deterministic
distribution. These operations are typed with the following rules, where
$\Tcons$ is the type constructor of probability distributions:
\begin{gather}
  \label{eq:nongraded-bind}
  \inferrule
  {\Gamma\vdash e_1:\Tcons\tau \\ \Delta,x:_{\text{\colorbox{gray!50}{$\infty$}}}\tau\vdash e_2:\Tcons\sigma}
  {\Delta+\Gamma\vdash \bbind x \leftarrow e_1; e_2:\Tcons\sigma}
  \\
  \label{eq:nongraded-return}
  \inferrule
  {\Gamma\vdash e:\tau}
  {\text{\colorbox{gray!50}{$\infty$}}\Gamma\vdash \breturn e:\Tcons\tau}
\end{gather}
The treatment of sensitivities in these rules is dictated by
differential privacy.  Roughly speaking, sensitivities measure the
privacy loss suffered by each input variable when the result of a
program is released.  Under this reading, the rule for $\bbind$ says
that the privacy loss of an input is the sum of the losses for each
sub-term.  Assuming that the bound variable $x$ has infinite
sensitivity in that rule is tantamount to imposing no restrictions on
its use.  This is possible thanks to the composition properties of
differential privacy: the result of a private computation is
effectively sanitized and can be used \emph{arbitrarily} without further
harming privacy. On the contrary, the scaling factor in the rule for
$\breturn$ implies that we cannot expect any privacy guarantees when
releasing the result of a computation on sensitive data---every input
is marked as having infinite privacy loss.  In practice, results of
differentially private computations must first be obfuscated with random
noise (cf. \cref{sec:interpretation-max-divergence}), and $\breturn$ is only used
with non-private inputs and the results of previous private
computations obtained by $\bbind$.

The above typing rules resemble those of standard monadic constructs, except for
the $\infty$ sensitivities. We could be tempted to combine interpretations for
the two sub-derivations in the $\bbind$ rule using a monad on metric spaces as
follows:
\begin{displaymath}
  \begin{tikzcd}
    \sem{\Delta} \ox \sem{\Gamma}
    \ar{r}[above = 1em]{\lambda\sem{e_2} \ox \sem{e_1}}
    & (\infty\sem{\tau} \multimap \sem{\Tcons\sigma}) \ox \sem{\Tcons\tau}
    \ar{r}[above = 1em]{?}
    & \sem{\Tcons\sigma}.
  \end{tikzcd}
\end{displaymath}
In a typical linear monadic calculus, we could just plug in the
internal Kleisli lifting in the morphism marked with ``?''. Here,
however, the types do not match up---there is the $\infty$ factor.

To model the probabilistic features of Fuzz, we need a structure that is similar
to a monad but with slightly different types for $\breturn$ and $\bbind$.  Our
solution lies in the notion of \emph{parameterized comonadic lifting}, which
refines the operations of a preexisting monad. In addition to the max divergence
originally used in Fuzz, we show how these liftings can be used to model the
statistical distance, which can be seen as measuring $\delta$ in
$(0, \delta)$-differential privacy. Later, we will generalize liftings to handle
$(\epsilon, \delta)$-differential privacy.

\newcommand{\Par}[1]{$#1$-parameterized }
\newcommand{\Weak}{weakly closed }

\subsection{Weakly Closed Monoidal Refinements}

The $\ox$ monoidal structure of $\Met$, which lies at the core of Fuzz's linear
analysis, is derived from the cartesian monoidal structure of $\Set$. These
categories are related by the forgetful functor $p$, which is strict monoidal.
Following \citet{DBLP:conf/popl/MelliesZ15}, we view $p$ as a \emph{refinement}
layering metrics on sets.
\begin{equation}
  \label{eq:metten}
  \xymatrix{
    (\Met,1,\ox) & & \Set \adjunction{ll}{\infty\cdot-}{p}
  }
\end{equation}

The above monoidal structures are closed, but exponentials only match up for
discrete metric spaces $\infty \cdot Z$---that is, $(p,p)$ is a map of
adjunctions of type
\begin{displaymath}
  (-\ox (\infty\cdot Z)\dashv (\infty\cdot Z)\multimap -)\arrow
  (-\times Z\dashv Z\Arrow -).
\end{displaymath}
Parameterized comonadic liftings are based on a generalization of this
situation:
\begin{definition}
  A {\em \Weak monoidal refinement} of a symmetric monoidal closed category
  (SMCC) $(\BB,\unit,\ox,\multimap)$ consists of a symmetric monoidal category
  $(\EE,\dunit,\dox)$ and an adjunction satisfying the following four conditions:
  \begin{equation}
    \label{eq:monref}
    \xymatrix{
      (\EE,\dunit,\dox) & & \BB \adjunction{ll}{L}{p}
    }
  \end{equation}
  \begin{enumerate}
  \item $p$ is strict symmetric monoidal and faithful;
    \label{r1}
  \item the unit of the adjunction is the identity;
    \label{r2}
  \item for each $X\in\BB$, $-\dox LX$ has a right adjoint
    $X\dpitchfork-$;
    \label{r3}
  \item for each $X\in\BB$, $(p,p)$ is a map of adjunction of type
    $(-\ox LX\dashv X\dpitchfork-)\arrow(-\ox X\dashv X\multimap
    -)$. \label{r4}
  \end{enumerate}
\end{definition}

There are many such refinements. Since $- \times (\infty \cdot X)$ is equal to
$- \ox (\infty \cdot X)$, it also has $\infty \cdot X \multimap -$ as a right
adjoint, which yields another example involving $\Met$:
\begin{equation}
  \label{eq:metprod}
  \xymatrix{
    (\Met,1,\times) & & \Set \adjunction{ll}{\infty\cdot-}{p}
  }
\end{equation}

\begin{comment}
These two examples can be generalized to arbitrary {\em
  $\CLat_\wedge$-fibrations} over $\Set$
\citep{DBLP:conf/cmcs/SprungerKDH18}. Via the Grothendieck
construction, such structures correspond to functors of type
$\Set^{op}\arrow\CLat_\wedge$, where the codomain is the category of
complete lattices and meet-preserving functions.  The forgetful
functor $p : \Met \to \Set$ is a $\CLat_\wedge$ fibration, and the
corresponding functor of type $\Set^{op} \to \CLat_\wedge$ maps a set
$X$ to the set $\Met_X$ of metrics over $X$, ordered pointwise by
$\ge$. Indeed, every $\CLat_\wedge$-fibration becomes \Weak monoidal
refinements of the CCC $\Set$ in two ways; the following theorem is a
consequence of \citet{topsym} (see also \citet[Section
2.3]{DBLP:conf/cmcs/SprungerKDH18}).
\begin{theorem}\label{thm:clat}
  Let $p:\EE\arrow\Set$ be a $\CLat_\wedge$-fibration. Then $p$ has a
  left adjoint $L$, $\EE$ has finite products $(\dot 1,\dtimes)$, and
  $\EE$ has a symmetric monoidal closed structure
  $(\dunit,\dox ,\dlarrow)$, such that the following are both \Weak monoidal
  refinements of the CCC $\Set$.
  \begin{displaymath}
    \xymatrix{
      (\EE,\dot 1,\dtimes) & & \Set \adjunction{ll}{L}{p}
      &
      (\EE,\dunit,\dox) & & \Set \adjunction{ll}{L}{p}
    }
  \end{displaymath}
\end{theorem}
\end{comment}
%
We'll see further examples in \cref{sec:rel} when extending Fuzz with
$(\epsilon,\delta)$-differential privacy.

Given a \Weak monoidal refinement as in
\eqref{eq:monref}, we write $!$ for the comonad $L\circ p$.
% The assumptions made on the adjunction reflect our programming
% language semantics.
% \begin{itemize}
% \item Since we are considering proof-irrelevant refinements, we
%   assume $p$ to be a faithful functor.
% \item The functor $p$ being strict symmetric monoidal means that the
%   monoidal structure on $\EE$ refines the one on $\BB$.
% \item The left adjoint $L$ maps an object in $\BB$ to its least
%   refinement in $\EE$.
% \end{itemize}
For $X,Y\in\EE$ and a morphism $f:pX\arrow pY$ in $\BB$, by $f:X\darrow Y$ we
mean that there exists a (necessarily unique) morphism $\dot f:X\arrow Y$ such that
$p\dot f=f$. Since the unit of the adjunction is the identity, we have 1)
$\fa{f\in\BB(X,pY)}f:LX\darrow Y$ and 2) $\fa{f\in\EE(LX,Y)}pf:LX\darrow Y$.

\subsection{Parameterized Comonadic Liftings}

Consider this simplified instance of the Fuzz $\bbind$ rule:
\begin{displaymath}
  \inferrule{
    \inferrule{ }{y:_1\Tcons\tau\vdash y:\Tcons\tau} \\
    \inferrule{\vdots}{\Delta,x:_\infty\tau\vdash e:\Tcons\sigma}
  }
  {\Delta,y:_1\Tcons\tau\vdash \bbind x\leftarrow y; e:\Tcons\sigma}
\end{displaymath}
Recall that $\bbind$ samples $x$ from $y \notin \Delta$ and computes $e$.
In a set-theoretic semantics that ignores sensitivities, $\Tcons$ might
correspond to the monad $D$ of discrete probability distributions on $\Set$, and
we can use the Kleisli lifting of $\sem{e}$ to interpret the entire derivation.
We can refine this interpretation with metrics by lifting $D$---that is, finding
a functor $\dot D$ such that $p \circ \dot D = D \circ p$---provided that the
following implication holds:
\begin{displaymath}
  f:\sem\Delta\ox!\sem\tau\nearrow\dot D\sem\sigma\implies
  \pklift f:\sem\Delta\ox\dot D\sem\tau\nearrow \dot D\sem\sigma
\end{displaymath}
The notion of parameterized comonadic lifting arises by generalizing this
argument to other monads.
\begin{definition}\label{def:inftylifting}
  A {\em \Par\dox $!$-lifting of $\mT$ along a \Weak monoidal
    refinement $p:\EE\arrow\BB$} is a mapping
  $\dot T:|\EE|\arrow|\EE|$ such that 1) $p\dot TX=TpX$, and 2) for
  any $X,Y,Z\in\EE$ and $\CC$-morphism $f$ such that
  $f:Z\dox {!}X\darrow \dot TY$, its parameterized Kleisli lifting
  (see \eqref{eq:parameterized-kleisli}) satisfies
  $\pklift f:Z\dox\dot TX\darrow\dot TY$.
\end{definition}
Every \Par\dox $!$-lifting $\dot T$ satisfies,  for any $X\in\EE$,
\begin{equation}
  \label{eq:unitne}
  \eta_{pX}: {!X} \darrow\dot TX,
\end{equation}
which allows us to extend $\dot T$ to a functor of type $\EE\arrow\EE$.

To model probability distributions in Fuzz, we pick a \Par\ox $!$-lifting
$\dot\dist$ of the distribution monad $\distmonad$.  The choice of $\dot\dist$
can vary---as we will soon see, one such $!$-lifting models the max divergence
in the original Fuzz.  A distribution type $\Tcons\tau$ is mapped to
$\dot\dist\sem{\tau}$, and $\breturn$ is interpreted using \eqref{eq:unitne} by
setting
\begin{align*}
  \sem{\breturn e} & : \sem{\infty \Gamma} \to \sem{\Tcons\tau} \\
  \sem{\breturn e}&\teq \eta_{p\sem\tau}\circ(\infty\cdot \sem e).
\end{align*}
Since $!$ commutes with $\ox$, the first factor is well-typed. We interpret
$\bbind$ as a parameterized Kleisli lifting:
\begin{align*}
  \sem{\bbind x\leftarrow e_1;e_2}&: \sem{\Delta + \Gamma} \to \sem{\Tcons\tau} \\
  \sem{\bbind x\leftarrow e_1;e_2}&\teq \pklift{\sem{e_2}}\circ(\sem{\Delta}\ox\sem{e_1})\circ \delta.
\end{align*}
To illustrate possible variations, the same interpretation works if we choose
$\dot \dist$ to be a \Par\times $!$-lifting of $\distmonad$: since
$Z\ox {!}X=Z\times {!}X$ and because $\id_{pX,pY}:X\ox Y\nearrow X\times Y$ is
non-expansive, any \Par\times lifting is also a \Par\ox lifting.  However, in
this case we can also strengthen the typing rule for $\bbind$ as follows:
\begin{equation}
  \label{eq:nongraded-with}
  \inferrule
  {\Gamma\vdash e_1:\Tcons\tau \\ \Gamma,x:_\infty\tau\vdash e_2:\Tcons\sigma}
  {\Gamma\vdash \bbind x \leftarrow e_1; e_2:\Tcons\sigma},
\end{equation}
where the context $\Gamma$ is shared between $e_1$ and $e_2$, yielding
an additive variant of the Fuzz bind rule (in the sense of linear logic).
%
% \jh{We don't seem to have non-graded examples of $\times$-strong
% assignments.  All the examples of $\times$-strong assignments we
% know are graded assignments. In that case, the $\Tcons$ in your
% return type would have an index accumulating under composition, so
% maybe $(2\epsilon, 2\delta)$ for your example.}
%
Observe that the domain of $\sem{e_2}$ is defined as
$\sem\Gamma\ox{!}  \sem\tau$, which is equal to $\sem\Gamma\times{!}
\sem\tau$. Thus we can apply the parameterized Kleisli lifting to $\sem{e_2}$.
The interpretation of $\bbind x\leftarrow e_1;e_2$ is then given as the
following composite:
\begin{displaymath}
  \sem{\bbind x\leftarrow e_1;e_2}=
  \xymatrix{
    \sem\Gamma \ar[r]^-{\langle \id,\sem{e_1}\rangle} & \sem\Gamma\times\dot \dist\sem{\tau} \ar[r]^-{\pklift{\sem{e_2}}} & \dot \dist \sem{\sigma}
  }.
\end{displaymath}
We will see examples of $\times$-parameterization in \cref{sec:grfuzz} after we
introduce graded liftings.

\subsection{Constructing Liftings via Parameterized Assignments}

% \jh{Here I start to get lost a bit on the motivation. This looks
% superficially like a slightly different presentation of
% !-liftings. Why is this notion interesting, like what does it show
% that was not shown by !-liftings?}  \sk{I rewrote this para so that
% I start with concrete things and arrive at an abstract formulation.}

\aa{I also feel that I do not fully grasp the motivation here. It seems that the
  main point is that the metrics on liftings do not actually depend on the
  metrics of their arguments.  Is this the case?}
\sk{
  I think relative monads (also ! lifting) covers the point.
  For me, assignment is a direct abstraction of Barthe-Olmedo
  condition (Theorem 1 of ``Beyond DP'' paper),
  strengthening data processing inequality. Overall,
  I wanted to connect
  various concepts developed around the DP work in a unified framework.
  Sequentially composable family  was also a part of this
  story, but I decided to drop it due to lack of space...
}

To express the max divergence with a comonadic lifting, we appeal to results of
\citet{BartheO13}. Their results are phrased in terms of a notion of
\emph{composability}, which we can recast as a sensitivity property.

\begin{lemma}
  Consider the situation \cref{eq:metten}.  Let
  $\Delta:|\Set|\arrow|\Met|$ be a mapping such that
  $p\Delta X=\dist X$. Then $\Delta$ satisfies the composability
  condition \cite{BartheO13}:
  \begin{displaymath}
    d_{\Delta Y}(\klift f \mu,\klift g \nu)\le d_{\Delta X}(\mu,\nu)+\sup_{x\in
      X}d_{\Delta Y}(f(x),g(x))
    % (X,Y\in\BB,\mu,\nu\in T X,f,g:X\arrow T Y)
  \end{displaymath}
  if and only if the {\em internalized Kleisli lifting} of $\mT$ is a
  non-expansive map
  $\kl_{X,Y}^\distmonad:(\infty\cdot X\multimap\Delta Y)\ox\Delta
  X\nearrow\Delta Y$.
\end{lemma}

This equivalent formulation can be readily generalized to other \Weak monoidal
refinements.
\begin{definition}\label{def:assign}
  A {\em \Par\dox assignment of $\EE$ on $\mT$} in a \Weak
  monoidal refinement $p:\EE\arrow\BB$ is a mapping
  $\Delta:|\BB|\arrow|\EE|$ such that $p\Delta X=TX$ and the
  internalized Kleisli lifting $\kl^\mT$ of $\mT$ in
  \cref{eq:internal-kleisli} satisfies
  \begin{equation}
    \label{eq:intkl}
    \kl_{X,Y}^\mT:(X\pitchfork \Delta Y)\dox\Delta X\darrow \Delta Y.
  \end{equation}
\end{definition}

Parameterized assignments turn out to be just an alternative presentation of
parameterized $!$-liftings---the former arising from concepts of probability
theory, and the latter mimicking the Fuzz typing rules. Formally, we have:
\begin{theorem}\label{th:eq}
  Consider a \Weak monoidal refinement \eqref{eq:monref} and a $\ox$-strong
  monad $\mT$ on $\BB$. There is an equivalence of preorders between
  \begin{enumerate}
  \item $\BLift\dox(\mT)$, the subpreorder of $\Ord(p,T\circ p)$
    consisting of \Par\dox $!$-liftings of $\mT$; and
  \item $\Asign\dox(\mT)$, the subpreorder of $\Ord(p,T)$ consisting
    of \Par\dox assignments of $\EE$ on $\mT$,
  \end{enumerate}
  where, given $F:A\arrow|\BB|$, $\Ord(p,F)$ is the class of mappings
  $\{G:A\arrow|\EE| ~|~pGX=FX\}$ ordered by
  $G\le G'\iff\fa{a\in A}\id_{Fa}:Ga\darrow G'a$.
\end{theorem}

\subsection{Max Divergence}
\label{sec:interpretation-max-divergence}

\aa{This could be a good place to explain the issue with $\infty$: having
  different points infinitely apart simplifies the rule for return}

Since the max divergence satisfies the composability condition \cite{BartheO13},
\Cref{th:eq} allows us to derive a corresponding parameterized lifting.  In
addition to the basic monadic operations, this lifting supports the real-valued
\emph{Laplace} distribution, a fundamental building block in differential
privacy. We can make a database query differentially private by adding Laplace
noise to its result while calibrating the scale of the noise according to the
query's sensitivity, as measured in terms of a suitable metric on databases. Its
density function is given by
\[
  L(\mu, b)(x) \teq \frac{1}{2b} \exp\left(- \frac{|x - \mu|}{b} \right),
\]
where $\mu \in \R$ and $b > 0$ are parameters controlling the mean and the scale
of the distribution.
The Laplace distribution induces a discrete distribution $\hat{L}(\mu, b) \in
D\R$ by truncating the sample up to some fixed precision (breaking ties
arbitrarily). This new distribution is compatible with the max divergence: it
satisfies the following max divergence bound:
\[
  \dMax_{\R}(\hat{L}(\mu, b), \hat{L}(\mu', b)) \leq \frac{|\mu - \mu'|}{b} .
\]
In other words, the mapping $\mu \mapsto \hat{L}(\mu, b)$ is a
$b^{-1}$-sensitive function from $\R$ to $\dist\R$ equipped with the
max divergence.\footnote{%
  This follows from $\epsilon$-privacy of the Laplace mechanism and stability of
max divergence under post-processing (see, e.g., \citet{DR14}).}
% \aa{Can/should we add a reference here?}

Fuzz exposes the Laplace distribution as a primitive
\[ \lapD[\epsilon] : {!_\epsilon} \R \multimap \Tcons\R , \] originally called
\textit{add\_noise}.  We can interpret this as:
\[
  \sem{\lapD[\epsilon]} \teq \lambda x.\, \hat{L}(x, 1/\epsilon) .
\]
The full Fuzz language also provides a type $\rset \tau$ of finite sets with
elements drawn from $\tau$, used to model sets of private data (``databases'').
This type is equipped with the Hamming distance, which is compatible with the
primitives operations on sets, e.g., computing the size, filtering according to
a predicate, etc. Extending the interpretation of Fuzz accordingly, a function
of type $!_{\epsilon}\rset \tau \multimap \Tcons \R$ corresponds to an
$\epsilon$-sensitive function from databases to distributions over $\R$ equipped
with the max divergence. As we discussed in the Introduction, this sensitivity
property is equivalent to $\epsilon$-differential privacy with respect to the
adjacency relation relating pairs of databases at Hamming distance at most $1$,
i.e., databases differing in at most one record.

\subsection{Statistical Distance}
\label{sec:interpretation-statistical-distance}

\citet{BartheO13} show that the composability condition is also valid for the
statistical distance:
\begin{displaymath}
  \dStat_X(\mu,\nu) \triangleq \frac 12\sum_{i\in X}|\mu(i)-\nu(i)|.
\end{displaymath}
This allows us to extend Fuzz with a new type constructor $\Tcons^{\dStat}$,
which we interpret using the statistical distance and its corresponding
lifting. In addition to $\breturn$ and $\bbind$, we can soundly incorporate a
primitive to compute the Bernoulli distribution, which models a biased coin
flip:
\begin{align*}
  \berD & : \R \multimap \Tcons^{\dStat} \B \\
  \sem{\berD}(p)(\btrue) & \teq \min(\max(p,0),1) \\
  \sem{\berD}(p)(\textbfalse) & \teq \min(\max(1-p,0),1).
\end{align*}
It is straightforward to check that the Bernoulli distribution satisfies the
following statistical distance bound:
\[
  \dStat_\B(\sem{\berD}(p), \sem{\berD}(p')) \leq |p - p'| ,
\]
implying that the type stated above is sound.

\section{Relations and $(\epsilon, \delta)$-Differential Privacy}
\label{sec:rel}

We will now shift gears and consider how to extend Fuzz to handle
$(\epsilon, \delta)$-differential privacy.  Recall that
$(\epsilon, \delta)$-privacy is a \emph{relational} property: a query $f$
satisfies the definition if it maps pairs of related input databases to related
output distributions, for suitable notions of ``relatedness.''  What makes this
notion challenging for Fuzz is that it cannot be phrased directly as a
sensitivity property (except for the special case $\delta = 0$, which we
analyzed above).  Rather than resorting to an entirely different verification
technique, we propose to incorporate relational reasoning into Fuzz by embedding
relations into metric spaces.

To warm up, we first show how to define $(\epsilon, \delta)$-differential
privacy in terms of a category of relations. Later (\cref{sec:grading}), we use
this formulation to capture the composition properties of differential privacy
with graded versions of parameterized liftings. Then, we consider how to
transfer these structures from relations to metric spaces via the \emph{path
  construction} (\cref{sec:path}).  Finally, we extend Fuzz with grading to
support relational properties (\cref{sec:grfuzz}).

\subsection{Differential Privacy in $\RSRel$}
% \subsection{Reflexive Symmetric Relations}
We begin by fixing a category of relations to work in. To smooth the eventual
transfer to metrics, which are reflexive and symmetric, we work with reflexive
and symmetric relations.

% \jh{Do we really need reflexive and symmetry, or could we take the
% reflexive/symmetric closure of a general relation as part of the
% path construction?}

\begin{definition}
  \label{def:rel-category}
  The category $\RSRel$ of reflexive, symmetric relations has as
  objects pairs $X = (|X|, {\sim_X})$ of a carrier set $|X|$ and a
  reflexive, symmetric relation ${\sim_X} \subseteq |X| \times
  |X|$. We will often use the carrier set $|X|$ to refer to $X$, and
  write $\sim$ when the underlying space is clear.
  A morphism $X \to Y$ is a function from $X$ to $Y$ that preserves the
  relation: $x \sim x' \implies f(x) \sim f(x')$.
  For $X,Y\in\RSRel$ and $f:|X|\arrow |Y|$, we write $f:X\rearrow Y$
  to mean $f\in\RSRel(X,Y)$.
\end{definition}
% \begin{figure}
%   \begin{center}
%     \begin{tabular}{ccc} \toprule Relation & Carrier Set & $a \sim b$ \\
%       \midrule
%     %
%       $1$ & $\{\star\}$ & true \\
%     %
%       $M X$ & $X$ & $a = b$ \\
%     %
%       $X_1 \times X_2$ &
%                          $X_1 \times X_2$ &
%                                             $\forall i\in\{1,2\}, a_i \sim_{X_i} b_i$ \\
%       $X_1 \ox X_2$ &
%                       $X_1 \times X_2$ & $\small\begin{array}{l}
%                                             (\pi_1 a=\pi_1 b\wedge \pi_2 a\sim_{X_2}\pi_2 b)\vee \\
%                                             (\pi_1 a\sim_{X_1}\pi_1 b\wedge \pi_2 a=\pi_2 b)
%                                          \end{array}$\\
%       \bottomrule
%     \end{tabular}
%   \end{center}
%   \caption{Basic constructions on $\RSRel$}
%   \label{fig:rel-objects}
% \end{figure}

The category $\RSRel$ has a terminal object, binary products and exponentials
($X,Y\in\RSRel$):
\begin{align*}
  1 &\teq (1,1\times 1)\\
  X\times Y &\teq (|X|\times |Y|, \{((x,y),(x',y'))~|~x\sim x',y\sim y'\})\\
  X\Arrow Y &\teq (|X|\Arrow |Y|,\{(f,f')~|~\forall{x\sim x'}.~f(x)\sim f'(x')\}).
\end{align*}
hence $\RSRel$ is a CCC.\footnote{%
  We note that $\RSRel$ has another symmetric monoidal closed structure with
tensor product $X\ox Y=(|X|\times |Y|,\{((x,y),(x,y'))~|~y\sim y'\}\cup
\{((x,y),(x',y))~|~x\sim x'\})$. The functor $q$ is also a \Weak monoidal
refinement of type $(\RSRel,1,\ox)\arrow\Set$.}
% \Cref{fig:rel-objects} summarizes basic constructions on this
% category, many of which have analogues in metric spaces.  Just like in
% $\Met$, there are two natural products. The first, $X_1 \times X_2$,
% requires pairs to be related pointwise; along with the obvious
% projections, it yields a categorical product on $\RSRel$.  The second,
% $X_1 \otimes X_2$, further requires that related pairs be equal in at
% least one of their components. Coupled with the terminal object $1$,
% these two products yield symmetric monoidal structure on $\RSRel$,
% which are also {\em closed}: both $X\ox-$ and $X\times-$ have right
% adjoints.
The forgetful functor $q:\RSRel\arrow\Set$, defined by $qX=|X|$ and
$qf=f$, has a left adjoint $M :\Set\arrow\RSRel$ endowing a set $X$
with the diagonal relation.  Moreover, $q$ strictly preserves the
cartesian closed structure, hence is a \Weak monoidal refinement of
the CCC $\Set$.
%for these two products:
\begin{align}
  % \label{eq:rsrelten}
  % \xymatrix{
  %   (\RSRel,1,\ox) & & \Set \adjunction{ll}{M}{q}
  % }
  % \\
  \label{eq:rsrelprod}
  \xymatrix{
    (\RSRel,1,\times) & & \Set \adjunction{ll}{M}{q}
  }
\end{align}
\begin{comment}
These are also derivable from the fact that $q$ is a
$\CLat_\wedge$-fibration and \cref{thm:clat}.
\end{comment}
%
The definition of differential privacy is parameterized by a set $db$ of
\emph{databases}, along with a binary \emph{adjacency} relation
$adj \subseteq db \times db$, which we assume to be symmetric and reflexive.
Conventional choices for $db$ include the set of sets of (or multisets, or
lists) of records from some universe of possible data, while $adj$ could relate
pairs of databases at symmetric difference at most $1$.  We recall the original
definition here for convenience.

\aa{This is a bit repetitive: we already talked about a general adjacency
  relation when discussing $\epsilon$-privacy.}

\begin{definition}[\citet{dwork2006our}]\label{def:indist}
  Let $\epsilon, \delta \in [0, \infty)$.  A randomized computation
  $f : db \to \dist X$ is $(\epsilon, \delta)$-\emph{differentially
    private} if for all pairs of adjacent databases $(d, d') \in adj$
  and subsets of outputs $S \subseteq X$, we have:
  \begin{align*}
    &f(d)(S) \leq \exp(\epsilon) \cdot f(d')(S) + \delta\quad\text{and}\\
    & f(d')(S) \leq \exp(\epsilon) \cdot f(d)(S) + \delta .
  \end{align*}
\end{definition}

We can track the privacy parameters by attaching the following {\em
  indistinguishability relation} to the \emph{codomain} of a
differentially private algorithm.  Given
$\epsilon,\delta\in[0,\infty)$ and a set $X$, we define
$\rDP(\epsilon, \delta)(X) \in \RSRel$ by setting
\begin{align*}
  \rDP(\epsilon,\delta)(X)\teq(\dist X,\{&(\mu,\nu) \mid \fa{S
                                           \subseteq X}\\
                                         &(\mu(S) \leq \exp(\epsilon) \cdot \nu(S) + \delta)
                                           \wedge\\
                                         &(\nu(S) \leq \exp(\epsilon) \cdot \mu(S) + \delta)\}) .
\end{align*}
\begin{proposition} \label{prop:dp-rel} A function
  $f : db \arrow \dist X$ is $(\epsilon,\delta)$-differentially
  private if and only if
  $f : (db, adj) \rearrow \rDP(\epsilon,\delta)(X)$.
\end{proposition}

Like $\epsilon$-differential privacy,
$(\epsilon, \delta)$-differential privacy behaves well under {\em
  sequential composition}.
\begin{theorem}[\citet{DMNS06}] \label{thm:dp-seq-comp} Let
  $f : db \to \dist X$ and $g : db \times X \to \dist Y$ be such that
  1) $f$ is $(\epsilon, \delta)$-differentially private, and 2)
  $g(-, x) : db \to \dist Y$ is $(\epsilon', \delta')$-differentially
  private for every $x \in X$.
  Then the composite function $d \mapsto g^\ddagger(d, f(d))$ is
  $(\epsilon + \epsilon', \delta + \delta')$-differentially private.
\end{theorem}

\section{Graded $!$-Liftings}
\label{sec:grading}

When reasoning about $\epsilon$-differential privacy in $\Met$ (and in Fuzz),
the privacy parameter is reflected in the scale of the \emph{domain} of a
non-expansive map, and the composition principle of $(\epsilon, 0)$-privacy
corresponds to composition of non-expansive maps. In $\RSRel$, there is no
analogous scaling operation for reasoning about $(\epsilon,\delta)$-privacy.  To
track these parameters through composition, we will instead use the query
\emph{codomain} by extending the liftings of \cref{sec:liftings} with
\emph{monoid grading} (analogously to graded extensions of monads
\citep{smirnov,mellies:BohmFestschrift,DBLP:conf/popl/Katsumata14}).  The
generality of these graded counterparts, which abstract the composition
properties of \cref{thm:dp-seq-comp}, will prove useful later on
(\cref{sec:other-divergences}), when modeling the composition behavior of other
relational properties.

We assume a \Weak monoidal refinement of a SMCC $(\BB,\unit,\ox,\multimap)$, as
in \eqref{eq:monref}, fixing a $\ox$-strong monad over $\BB$, and a preordered
monoid $(M,\le,1,\cdot)$.
\begin{definition}
  An {\em $M$-graded \Par\dox $!$-lifting} is a monotone function
  $\dot T:(M,\le)\arrow\Ord(p,T\circ p)$ (cf. \cref{th:eq}) such that the
  parameterized Kleisli lifting \cref{eq:parameterized-kleisli} of $\mT$
  satisfies
  \begin{displaymath}
    f:Z\dox {!}X\darrow \dot T\alpha Y\implies \pklift f:Z\dox\dot T\beta X\darrow\dot T(\beta\cdot \alpha)Y.
  \end{displaymath}
\end{definition}
When $M=1$, this definition reduces to its non-graded counterpart. Monotonicity
of an $M$-graded \Par\dox $!$-lifting $\dot T$ means that for any monoid element
$\alpha\le\beta$ and any object $X\in\EE$, we have
$\dot T\alpha X\le\dot T\beta X$. Regarding the unit,
$\eta_{pX}:{!}X\darrow\dot T\alpha X$ holds for any $\alpha\in M$ and $X\in\EE$.
From this, each $\dot T\alpha$ extends to an endofunctor over $\EE$.
\begin{definition}
  An {\em $M$-graded \Par\dox assignment of $\EE$ on $\mT$} is a
  monotone function $\Delta:(M,\le)\arrow\Ord(p,T)$ such that the
  internalized Kleisli lifting morphism $\kl^\mT$ of $\mT$ in
  \cref{eq:internal-kleisli} satisfies
  \begin{equation}
    \kl_{X,Y}^\mT:(X\pitchfork \Delta \alpha Y)\dox\Delta \beta X\darrow \Delta (\beta\cdot \alpha)Y.
  \end{equation}
\end{definition}
Once again, the original definition corresponds to the case $M=1$. To illustrate
this notion, in the \Weak monoidal refinement \eqref{eq:metten}, an
$M$-graded \Par\ox assignment $\Delta$ of $\Met$ on $\distmonad$ consists of a
family of metrics $d_{\Delta\alpha X}$ on $\dist X$, indexed by $\alpha\in M$,
such that, for any $X,Y\in\Set,\mu,\nu\in \dist X,f,g:X\arrow \dist Y$ and
$\alpha,\beta\in M$, we have
\begin{displaymath}
  d_{\Delta(\alpha\cdot\beta)Y}(f^\dagger \mu,g^\dagger \nu)\le d_{\Delta\alpha X}(\mu,\nu)+\sup_{x\in X}d_{\Delta\beta Y}(f(x),g(x)).
\end{displaymath}
In this case, assignments encode a family of distances on distributions enjoying
the sequential composition theorem for statistical divergences proposed by
\citet[Theorem 1]{BartheO13} (see also \citet{OlmedoThesis}).

An instance of a graded assignment is the indistinguishability
relation for differential privacy. The following is a consequece of
\cref{thm:dp-seq-comp}:
\begin{proposition}\label{pp:rdpcomp}
  Let $\R_{\ge 0}^{+}$ be the additive monoid on the set $\R_{\ge 0}$
  of nonnegative real numbers. In the \Weak monoidal refinement
  \eqref{eq:rsrelprod}, the indistinguishability relation, regarded as
  a mapping of type $\rDP:\R_{\ge 0}\times \R_{\ge 0}\arrow |\RSRel|$,
  is an $(\R_{\ge 0}^{+}\times\R_{\ge 0}^{+})$-graded \Par\times
  assignment of $\RSRel$ on $\distmonad$.
\end{proposition}

% \jh{Here again the difference between this definition and the
% previous ones seems a bit minor? I guess we can say this is somehow
% given by the composition theorem in differential privacy. But we
% would need to explain that connection a lot more. Also I see that
% here we require cartesian, but I don't see the importance of this or
% where it comes from (definitely in differential privacy we are not
% talking about cartesian structure :)}

We have a graded analogue of \cref{th:eq}.
% Moreover, when the monoidal structure $\dox$ on $\EE$ is cartesian, the
% sequentially composable families are equal to \Par\dox assignments.
%
% AAA We are not talking about sequentially composable families anymore.
\begin{theorem}\label{th:grdeq}
  Consider a \Weak monoidal refinement \eqref{eq:monref} and a $\ox$-strong
  monad $\mT$ on $\BB$.  Let $M$ be a preordered monoid. The following preorders
  are equivalent:
  \begin{enumerate}
  \item $\BLift\dox(\mT,M)$, the preordered class of
  $M$-graded \Par\dox $!$-liftings of $\mT$ with the pointwise preorder, and
\item $\Asign\dox(\mT,M)$, the preordered class of
  $M$-graded \Par\dox assignments on $\mT$ with the pointwise preorder.
  \end{enumerate}
\end{theorem}

\section{Transfers of Assignments}
\label{sec:path}

% \jh{I would move the rest of this section to later in the paper
% (maybe the current section 8?) after doing a big example of
% interpreting terminating Fuzz, I think it was better
% there. Otherwise this section feels very heavy to me, we are
% introducing a lot of concepts but it is hard to see how they all
% work together. The path construction and assignment transfer stuff
% is neat but it seems a bit orthogonal to the lifting stuff. Plus, we
% do not have examples of path construction until after terminating
% Fuzz, anyways.}  \sk{I agree, I have moved like this.}

So far, we are able to model relational properties and their composition
behavior in $\RSRel$. In this section, we will show how to carry out this
reasoning in a different category, namely $\Met$. We introduce this idea
abstractly, then give a concrete example called the \emph{path construction} for
the special case of $\RSRel$ and $\Met$.  In our framework, this transfer of
structure is induced by a morphism between \Weak monoidal refinements.
\begin{definition}
  Consider two \Weak monoidal refinements of a
  SMCC $\BB$, and a functor $F:\EE\arrow\FF$:
  \begin{equation}
    \label{eq:mapref}
    \xymatrix{
      (\EE,\dunit,\dox) \ar[rr]^-{F} & & (\FF,\ddunit,\ddox) \\
      & \BB \adjunction{lu}{L}{p} \adjunction{ru}{L'}{p'}
    }
  \end{equation}
  $F$ is a {\em morphism of \Weak monoidal refinements} if
  \begin{enumerate}
  \item $F$ is strict symmetric monoidal,
  \item $(\mathrm{Id}_\BB,F):(L\dashv p)\arrow(L'\dashv p')$, and
  \item
    $(F,F):(-\dox LX\dashv X\dpitchfork-)\arrow(-\ddox L'X\dashv
    X\ddpitchfork-)$ for each $X\in\BB$.
  \end{enumerate}
  We write
  $F:(\EE,\dunit,\dox,L,p)\arrow_\BB (\FF,\ddunit,\ddox,L',p')$.
\end{definition}

\begin{theorem}\label{th:mapasign}
  If $F:\EE\arrow\FF$ is a morphism of \Weak monoidal refinements, then
  $F\circ -$ restricts to a monotone function of type
  $\Asign\dox(\mT,M)\arrow\Asign{\ddox}(\mT,M)$.
\end{theorem}

We express $(\epsilon,\delta)$-privacy in metric spaces through the {\em path
  construction functor} $P:\RSRel\arrow\Met$.  Given an object $X \in \RSRel$,
we define a metric on the underlying set by counting the number of times
$\sim_X$ must be composed to relate two points. Such metrics are also known as
\emph{path metrics}, and the corresponding metric spaces are known as
\emph{path-metric spaces}.

\begin{definition}
  \label{def:path-metric} Let $X \in \RSRel$. The \emph{path metric}
  is a metric on $X$ defined as follows: $d(x, x')$ is the length $k$
  of the shortest path of elements $x_0, \ldots, x_k$ such that
  $x_0 \teq x$, $x_k \teq x'$, and $x_i \sim x_{i+1}$ for every
  $i \in \{0,\ldots,k-1\}$. If no such sequence exists, we set
  $d(x, x') \teq \infty$.  We write $PX$ for the corresponding metric
  space. This definition can be extended to a functor
  $\RSRel \to \Met$ that acts as the identity on morphisms.
\end{definition}

Conversely, we can turn any metric space into an object of $\RSRel$ by relating
elements at distance at most $1$.

\begin{definition}[At most one]
  \label{def:at-most-one}
  Given $X \in \Met$, we define $QX \in \RSRel$ by
  \begin{displaymath}
    QX=(X,\{(x,x')~|~d(x,x')\le 1\}), \quad Qf=f.
  \end{displaymath}
% as follows.  Its
%   carrier is the same as the carrier of $X$, and $x \sim x'$ if and
%   only if $d(x, x') \leq 1$; this relation is reflexive because
%   $d(x, x) = 0$, and it is symmetric because of the symmetry of the
%   metric.
  % This definition extends to a functor $Q : \Met \to \RSRel$
  % that behaves as the identity on morphisms.  Indeed, given a
  % non-expansive function $f : X \to Y$, if $x$ and $x'$ in $X$ are
  % such that $d(x, x') \leq 1$, then
  % $d(f(x), f(x')) \leq d(x, x') \leq 1$, showing that $f$ preserves
  % the ``at-most-one'' relation.
\end{definition}
\begin{theorem}\label{th:pathfun}
  The functor $P:\RSRel\arrow\Met$ is fully faithful, and a left adjoint to
  $Q:\Met\arrow\RSRel$.
\end{theorem}
\begin{theorem}\label{th:transpath}
  The path metric functor $P:\RSRel\arrow\Met$ is a morphism of
  \Weak monoidal refinements:
  \begin{align*}
%    &P:(\RSRel,1,\ox,M,q)\arrow_\Set(\Met,1,\ox,\infty\cdot-,p), \\
    &P:(\RSRel,1,\times,M,q)\arrow_\Set(\Met,1,\times,\infty\cdot-,p).
  \end{align*}

  % \begin{equation}
  %   \xymatrix{
  %     (\RSRel,1,\ox) \ar[rr]^-{P} & & (\Met,1,\ox) \\
  %     & \Set \adjunction{lu}{M}{q} \adjunction{ru}{\infty\cdot-}{p}
  %   }
  % \end{equation}
  % \begin{equation}
  %   \xymatrix{
  %     (\RSRel,1,\times) \ar[rr]^-{P} & & (\Met,1,\times) \\
  %     & \Set \adjunction{lu}{M}{q} \adjunction{ru}{\infty\cdot-}{p}
  %   }
  % \end{equation}
\end{theorem}
% So, graded \mbox{\Par\ox} and \mbox{\Par\times} assignments of
% $\RSRel$ can be transferred to corresponding parameterised
% assignments of $\Met$ by the path metric functor $P$.

% \begin{theorem} \label{thm:assign-rel-met} Let $\star$ be the
%   symmetric monoidal structure $\ox$ or the cartesian monoidal
%   structure $\times$ on $\RSRel$. Let $\Delta$ be a $G$-graded
%   $\star$-strong assignment of $\RSRel$ on $\distmonad$.  Then the
%   family $P(\Delta \alpha X)$, indexed by $\alpha \in G$ and
%   $X\in\Set$, is a $G$-graded $\star$-strong assignment of $\Met$ on
%   $\distmonad$.
% \end{theorem}

\section{Metric Semantics of Graded Fuzz}
\label{sec:grfuzz}

We have all the ingredients we need to extend Fuzz. We fix a preordered monoid
$(M,\le,1,\cdot)$, augment Fuzz with $M$-graded monadic types
$\Tcons_\alpha \tau$, adjust the typing rules \eqref{eq:nongraded-return} and
\eqref{eq:nongraded-bind} to track the grading, and add a grading subsumption
rule.
\begin{gather}
  \label{eq:gradedsub}
  \inferrule
  {\Gamma\vdash e:\Tcons_\alpha\tau \\ \alpha\le\beta}
  {\Gamma\vdash e:\Tcons_{\beta}\tau}
  \\
  \label{eq:gradedeta}
  \inferrule
  {\Gamma\vdash e:\tau}
  {\infty\cdot\Gamma\vdash\breturn e: \Tcons_\alpha\tau}
  \\
  \label{eq:gradedbind}
  \inferrule
  {\Gamma\vdash e_1:\Tcons_\alpha\tau \\ \Delta,x:_\infty\tau\vdash e_2:\Tcons_{\beta}\sigma}
  {\Delta+\Gamma \vdash \bbind x\leftarrow e_1;e_2:\Tcons_{\alpha\cdot\beta}\sigma}
\end{gather}
Here, $\alpha$ and $\beta$ range over $M$. The interpretation in
\cref{sec:liftings} can be easily extended when we have an $M$-graded
\mbox{\Par\ox} $!$-lifting $\dot\dist$ of the distribution monad $\distmonad$
along $p:\Met\arrow\Set$, taking
$\sem{\Tcons_\alpha\tau}=\dot\dist\alpha\sem \tau$.

If instead we have a $M$-graded \mbox{\Par\times} $!$-lifting $\dot\dist$ of
$\distmonad$ along $p:\Met\arrow\Set$, we can again interpret the monadic type
with $\sem{\Tcons_\alpha\tau}=\dot\dist\alpha\sem\tau$. However, we can replace
\eqref{eq:gradedbind} with a stronger rule for $\bbind$:
\begin{equation}
  \label{eq:gradedbindcart}
  \inferrule
  {\Gamma\vdash e_1:\Tcons_\alpha\tau \\ \Gamma,x:_\infty\tau\vdash e_2:\Tcons_{\beta}\sigma}
  {\Gamma \vdash \bbind x\leftarrow e_1;e_2:\Tcons_{\alpha\cdot\beta}\sigma}
\end{equation}

% In what follows, we construct graded parameterized assignments in $\RSRel$ by
% appealing to composition properties for
% divergences~\citep{BartheO13,OlmedoThesis}.  By applying the path construction,
% we can transfer graded parameterized assignments of $\RSRel$ into those of
% $\Met$, supporting new graded distribution types and typing rules for Fuzz.

\subsection{Modeling $(\epsilon, \delta)$-Differential Privacy in Graded Fuzz}

We have seen that differential privacy can be expressed in $\RSRel$, with the
sequential composition property providing a graded assignment that can be
transferred to $\Met$ through the path metric.  By \cref{pp:rdpcomp} and
\cref{th:grdeq}, in the \Weak monoidal refinement \eqref{eq:rsrelprod} we have
$ \rDP
%  \in\Comp(\distmonad,\R_{\ge 0}^{+}\times\R_{\ge 0}^{+})
  \in\Asign\dtimes(\distmonad,\R_{\ge 0}^{+}\times\R_{\ge 0}^{+}).
$
From \cref{th:mapasign}, the path construction maps $\rDP$ to a graded
\Par\dtimes assignment of $\Met$ on $\distmonad$:
\begin{displaymath}
  P\circ\rDP\in\Asign\dtimes(\distmonad,\R_{\ge 0}^{+}\times\R_{\ge 0}^{+})
\end{displaymath}
in the \Weak monoidal refinement \eqref{eq:metprod}.  Below we identify
$P\circ\rDP$ as a $\R_{\ge 0}^{+}\times\R_{\ge
  0}^{+}$-graded \Par\dtimes $!$-lifting of $\distmonad$.
By posing
\[
  \sem{\Tcons_{(\epsilon, \delta)} \tau} \triangleq P(\rDP(\epsilon,\delta)(\sem{\tau})) ,
\]
% where $p : \Met \to \Set$ is the forgetful functor,
the following specializations of \eqref{eq:gradedsub},
\eqref{eq:gradedeta}, and \eqref{eq:gradedbindcart} are sound:
\begin{gather*}
  % \label{eq:graded-dp}
  \inferrule
  {\Gamma\vdash e:\Tcons_{(\epsilon, \delta)}\tau \\
    \epsilon \leq \epsilon' \\ \delta \leq \delta'} {\Gamma\vdash
    e:\Tcons_{(\epsilon', \delta')}\tau}
  \\
  \inferrule {\Gamma\vdash e:\tau} {\infty\cdot\Gamma\vdash\breturn e:
    \Tcons_{(\epsilon, \delta)}\tau}
  \\
  \inferrule
  {\Gamma\vdash e_1:\Tcons_{(\epsilon, \delta)}\tau \\
    \Gamma,x:_\infty\tau\vdash e_2:\Tcons_{(\epsilon',
      \delta')}\sigma} {\Gamma \vdash \bbind x\leftarrow
    e_1;e_2:\Tcons_{(\epsilon + \epsilon', \delta + \delta')}\sigma}
\end{gather*}
Much like we did for the standard $!$-liftings in
\cref{sec:interpretation-max-divergence,sec:interpretation-statistical-distance},
we can add primitive distributions to the system. The basic building block for
$(\epsilon, \delta)$-differential privacy is the real-valued \emph{Gaussian} or
\emph{normal} distribution. Given a mean $\mu \in \R$ and a variance
$\sigma \in \R$, this distribution has density function
\[
  N(\mu, \sigma)(x) \triangleq \frac{1}{\sqrt{2 \pi \sigma^2}} \exp\left( \frac{(x - \mu)^2}{2 \sigma^2} \right) .
\]
A result from the theory of differential privacy states that if we have a
numeric query $q : db \to \R$ whose results differ by at most $1$ on adjacent
databases, then adding noise drawn from $N(0, \sigma)$ to the query's result
yields an $(\epsilon, \delta)$-differentially private algorithm as long as
$\sigma \geq s(\epsilon, \delta) \triangleq 2 \ln(1.25/\delta) / \epsilon$ (see,
e.g., \citet[Theorem A.1]{DR14}). Like we did for the Laplace distribution, we
may discretize the result of the noised query to any fixed precision while
preserving privacy, yielding a distribution we call $\hat{N}(\mu, \sigma)$. The
function $\lambda x.\, \hat{N}(x, \sigma) : \R \to \dist \R$ can then be
interpreted in $\RSRel$:
\[
  \lambda x.\, \hat{N}(x, s(\epsilon, \delta)) : Q\R \rearrow \rDP(\epsilon, \delta)(\R)
\]
Unfolding definitions, $Q\R$ relates pairs of real numbers that are at most $1$
apart under the standard Euclidean distance. The path construction gives a
non-expansive map between path-metric spaces:
\[
  \lambda x.\, \hat{N}(x, s(\epsilon, \delta)) : PQ\R \nearrow P(\rDP(\epsilon, \delta)(\R))
\]
The metric space $PQ\R$ is $\R$ but with the metric rounded up to the nearest
integer; we introduce a corresponding Fuzz type $\lceil \R \rceil$. Then, we can
introduce a new Fuzz term $\gauD[\epsilon, \delta]$ for $\epsilon, \delta > 0$
with type
\[
  \gauD[\epsilon, \delta] : \lceil \R \rceil \multimap \Tcons_{(\epsilon, \delta)} \R
\]
and the interpretation
\[
  \sem{\gauD[\epsilon, \delta]} \triangleq \lambda x.\, N(x, s(\epsilon, \delta)) .
\]

\paragraph*{Typing $(\epsilon, \delta)$-Differential Privacy}
We can now capture $(\epsilon, \delta)$-privacy via Fuzz types. Consider the
judgment:
\[
  \vdash e : \db \multimap \Tcons_{(\epsilon, \delta)} \tau
\]
where $\db$ is interpreted as the path-metric space $P(db, adj)$. Note that if
$\db = \rset \sigma$ is the space of sets of $\sigma$ and we take the Hamming
distance as the metric $d_DB$ on this space (as we did in previous examples),
then $(db, d_{DB})$ is automatically a path metric space for the relation
relating any two databases at Hamming distance at most $1$.  We have a
non-expansive map
\[
  \sem{e} : P(db, adj) \nearrow P(\rDP(\epsilon, \delta)(\sem{\tau}))
  .
\]
Since the path functor $P$ is full and faithful
(\cref{th:pathfun}), we have a relation-preserving map
\[
  \sem{e} : (db, adj) \rearrow \rDP(\epsilon, \delta)(\sem{\tau})
\]
in $\RSRel$.  By \cref{prop:dp-rel}, this map satisfies
$(\epsilon, \delta)$-privacy.
We give two examples to demonstrate the type system.  Consider the type $db
\multimap \R$, typically used to model $1$-\emph{sensitive queries}. Applying
$PQ$, we find that we can model $1$-sensitive queries with the type $\lceil db
\rceil \multimap \lceil \R \rceil$. Now $\lceil db \rceil$ rounds up the metric
on $db$ to the nearest integer, but since $db$ is already a path metric space,
$\lceil db \rceil$ and $db$ have the same denotations. Thus, $1$-sensitive
queries can be interpreted as type $\db \multimap \lceil \R \rceil$.

Let $q_1$ and $q_2$ be $1$-sensitive queries of type
$\db \multimap \lceil \R \rceil$, and consider the program
$\mathit{two}\_q$:
\begin{align*}
  \lambda db.\, &\bbind a_1 \leftarrow \gauD[\epsilon, \delta](q_1(db));\\
                &\bbind a_2 \leftarrow \gauD[\epsilon, \delta](q_2(db));\\
                &\breturn (a_1 + a_2)
\end{align*}
This program evaluates the first query $q_1$ and adds Gaussian noise to the
answer, evaluates the second query $q_2$ and adds more Gaussian noise, and
finally returns the sum of the two noisy answers. By applying the typing rules
for the Gaussian distribution, along with the graded monadic rules, we can
derive the following type:
\[
  \vdash \mathit{two}\_q : \db \multimap \Tcons_{(2\epsilon, 2\delta)}
  \R
\]
Though the database $db$ is used twice in the program, it has sensitivity $1$ in
the final type. This accounting follows from the bind rule, which allows the
contexts of its premises to be shared. The fact that the database is used twice
is instead tracked through the grading on the codomain, where the privacy
parameters $(\epsilon, \delta)$ sum up. By soundness of the type system,
$\mathit{two}\_q$ is $(2\epsilon, 2\delta)$-differentially private.

Other types can capture variants of differential privacy. For example,
suppose added the queries first and noised just once:
\[
  \lambda db.\, \blet s \leftarrow q_1(db) + q_2(db); \gauD[\epsilon,
  \delta](s)
\]
We use standard syntactic sugar for let bindings; call this program
$\mathit{two}\_q'$. We can derive the following type:
\[
  \vdash \mathit{two}\_q' : {!_2} \db \multimap \Tcons_{(\epsilon,
    \delta)} \R
\]
This type is not equivalent to the type for
$\mathit{two}\_q$. However, we can still interpret it in terms of
differential privacy.  In general, consider the following judgment:
\[
  \vdash e : {!_2} \db \multimap \Tcons_{(\epsilon, \delta)} \tau
\]
By soundness, the interpretation is non-expansive:
\[
  \sem{e} : 2 \cdot P(db, adj) \nearrow P(\rDP(\epsilon,
  \delta)(\sem{\tau})) .
\]
Though the scaling of a path metric is not necessarily a path metric,
we can still give a meaning to this judgment. For any two input
databases $(x, x') \in adj$, the distance between $\sem{e}x$ and
$\sem{e}x'$ in $P(\rDP(\epsilon, \delta)(\sem{\tau}))$ is at most
$2$. Suppose that the distance is exactly $2$ (smaller distances yield
stronger privacy bounds). Then, there must exist an intermediate
distribution $y \in \dist |\sem{\tau}|$ such that
\[
  \sem{e}x \sim_{\rDP{(\epsilon, \delta)}(\sem{\tau})} y
  \sim_{\rDP{(\epsilon, \delta)}(\sem{\tau})} \sem{e}x' .
\]
Unfolding definitions, a small calculation shows that the output
distributions must be related by
\[
  \sem{e}x \sim_{\rDP(2\epsilon, (1 +
    \exp(\epsilon))\delta)(\sem{\tau})} \sem{e}x' .
\]
Hence we have a relation-preserving map
\[
  \sem{e} : (db, adj) \rearrow \rDP(2\epsilon, (1 +
  \exp(\epsilon))\delta)(\sem{\tau})
\]
and by \cref{prop:dp-rel}, the map $\sem{e}$ and our program $\mathit{two}\_q'$
satisfy $(2\epsilon, (1 + \exp(\epsilon))\delta)$-differential privacy.

\subsection{Modeling Other Divergences}
\label{sec:other-divergences}

% \sk{I think we are not limited to $f$-divergences.}

% \sk{I saw several $\times$-strong assignments of $\RSRel$ - how to
% prove \eqref{eq:times-strong-rsrel} for various graded divergences?
% }

The $(\epsilon, \delta)$-differential privacy property belongs to a broader
class of probabilistic relational properties: two related inputs lead to two
output distributions that are a bounded distance apart, as measured by some
divergence on distributions. By varying the divergence,
% \sknote{What is the definition of divergence?}
% \jh{I'm being sketchy here. Essentially we can do the composable divergences.}
these properties can capture different notions of probabilistic sensitivity.
To support a grading, the divergences must be composable in a certain sense.

\begin{definition} \label{def:graded-comp} Let
  $H=(\R^\infty_{\ge 0},\le,u,\bullet)$ be a partially ordered monoid
  over the non-negative extended reals. A family of
  divergences $d_X$ on $\dist X$ indexed by sets $X$ is {\em
    $H$-composable} if for any $f,g:X\arrow\dist Y$ and
  $\mu,\nu\in\dist X$, we have
  \begin{displaymath}
    d_Y(f^\dagger(\mu),g^\dagger(\nu))\le
    d_X(\mu,\nu) \bullet \sup_{x\in X}d_Y(f(x),g(x)).
  \end{displaymath}
\end{definition}

Previously, \citet{BartheO13} proposed \emph{weak} and \emph{strong
composability} to study sequential composition properties for the class of
$f$-divergences. (The skew divergence in $(\epsilon, \delta)$-differential
privacy is an example of an $f$-divergence.) These notions coincide when working
with full distributions rather than sub-distributions, as in our settings. Given
any family of composable divergences, we can build a corresponding {\em graded}
$!$-lifting of $\RSRel$ on $\distmonad$.

\begin{theorem} \label{thm:graded-comp-assign} Let $d_X$ be an
  $H$-composable family of divergences on $\dist X$ and
  $q:\RSRel\arrow\Set$ be the forgetful functor. Define a mapping
  $R(d)$ by
  \begin{displaymath}
    R(d)(\delta)(X)\teq
    (\dist X,\{(\mu,\nu)~|~d_X(\mu,\nu)\le\delta, d_X(\nu,\mu)\le\delta\}).
  \end{displaymath}
  Then $R(d)$ is a monotone mapping of type
  $(\R^\infty_{\ge 0},\le)\arrow \Ord(q,\dist)$, and is an
  $H$-graded \Par\times assignment of $\RSRel$ on $\distmonad$.
\end{theorem}
As usual, we identify $R(d)$ and the corresponding
$H$-graded \Par\times $!$-lifting of $\distmonad$ along
$q:\RSRel\arrow\Set$.

% \begin{proof}
%   It is evident that $(\alpha,\delta)\le(\alpha',\delta')$ implies
%   ${\sim}_{R(d)(\alpha,\delta)(X)}\subseteq
%   {\sim}_{R(d)(\alpha',\delta')(X)}$. We thus show that
%   ${\sim}_{R(d)(\alpha',\delta')(X)}$ satisfies
%   \eqref{eq:times-strong-rsrel}. Suppose that $f,g:X\arrow \dist Y$
%   and $\mu,\nu\in\dist X$ satisfies
%   $\mu\sim_{R(d)(\alpha,\delta)(X)}\nu$ and
%   $\fa{x\in X}f(x)\sim_{R(d)(\beta,\delta')(Y)}g(x)$. Equivalently,
%   $d^\alpha_X(\mu,\nu)\le\delta$ and
%   $\sup_{x\in X}d^\beta_Y(f(x),g(x))\le\delta'$.  We then obtain the
%   following from the graded composability:
%   \begin{displaymath}
%     d^{\alpha\cdot\beta}_Y(f^\dagger(\mu),g^\dagger(\nu))
%     \le
%     m(d^\alpha_X(\mu,\nu),\sup_{i\in I}d^\beta_Y(f(i),g(i)))
%     \le
%     m(\delta,\delta')
%   \end{displaymath}
%   which is equivalent to
%   $f^\dagger(\mu)\sim_{R(d)(\alpha\cdot\beta,m(\delta,\delta'))}g^\dagger(\nu)$.
%   This ends the proof of \eqref{eq:times-strong-rsrel}.
% \end{proof}

We extend Fuzz with two examples of composable divergences, briefly sketching
the composition rules, graded assignment structure, and Fuzz typing rules. We
present further examples in \cref{app:other-divegences}.

\paragraph*{KL Divergence}
The \emph{Kullback-Leibler (KL)} divergence, also known as
\emph{relative entropy}, measures the difference in information
between two distributions. For discrete distributions over $X$, it is
defined as:
\[
  \dKL_X(\mu, \nu) \teq \sum_{i \in X} \mu(i) \log
  \frac{\mu(i)}{\nu(i)} ,
\]
where summation terms with $\mu(i) = \nu(i) = 0$ are defined to be $0$, and
terms with $\mu(i) > \nu(i) = 0$ are defined to be $\infty$. This divergence is
reflexive, but it is not symmetric and does not satisfy the triangle
inequality. (It is not immediately obvious, but the KL divergence is always
non-negative.) We can define a family of relations that models pairs of
distributions at bounded KL divergence. For any $\alpha \in \R$ and set $X$, we
define:
\begin{align*}
  & \rKL(\alpha)(X) \teq R(\dKL)(\alpha)(X)\\
  & \quad = ( \dist X, \{ (\mu, \nu)
    \mid \dKL_X(\mu, \nu) \leq \alpha, \dKL_X(\nu, \mu) \leq \alpha \} )
\end{align*}
Note that $\alpha$ need not be an integer---it can be any real number.  This
relation is reflexive and symmetric, hence an object in
$\RSRel$. \citet[Proposition 5]{BartheO13} show that $\dKL$ is $H$-composable
for $H = (\R, \leq, 0, +)$, so $\rKL(\alpha)(X)$ is a $H$-graded \Par\times
assignment of $\RSRel$ on $\distmonad$ by \cref{thm:graded-comp-assign}. By
applying the path construction, we get a $H$-graded \Par\times assignment of
$\Met$ on $\distmonad$ which we can use to capture KL divergence as a graded
distribution type:
\[
  \sem{\Tcons^{\dKL}_\alpha \tau} \teq P(\rKL(\alpha)(\sem{\tau}))
\]
For $\bbind$, for instance, we obtain the following typing rule:
\[
  \inferrule
  {\Gamma\vdash e_1:\Tcons^{\dKL}_{\alpha}\tau \\
    \Gamma,x:_\infty\tau\vdash e_2:\Tcons^{\dKL}_{\beta}\sigma}
  {\Gamma \vdash \bbind x\leftarrow e_1;e_2:\Tcons^{\dKL}_{\alpha +
      \beta}\sigma}
\]
Like we did for the max divergence of differential privacy, we can
also introduce primitive distributions and typing rules into Fuzz. For
instance, a standard fact in probability theory is that the standard
Normal distribution satisfies the bound
\[
  \dKL_{\R} (N(\mu_1, 1), N(\mu_2, 1)) \leq (\mu_1 - \mu_2)^2
\]
If we again discretize this continuous distribution to
$\hat{N}(\mu, 1)$ and interpret the primitive term
$\sem{\norD} = \lambda x.\, \hat{N}(x, 1)$, the following typing rule
is sound:
\[
  \inferrule { } {\Gamma \vdash \norD : \lceil \R \rceil \multimap
    \Tcons^{\dKL}_{1} \R}
\]

\paragraph*{$\chi^2$ Divergence}
For discrete distributions on $X$, the $\chi^2$ divergence is defined as
\[
  \dChi_X(\mu, \nu) \teq \sum_{i \in X} \frac{(\mu(i) -
    \nu(i))^2}{\nu(i)} ,
\]
where summation terms with $\mu(i) = \nu(i) = 0$ are defined to be
$0$, and terms with $\mu(i) > \nu(i) = 0$ are defined to be $\infty$.
Note that this divergence is not symmetric and does not satisfy the
triangle inequality. We define a family of reflexive symmetric
relations that models pairs of distributions at bounded
$\chi^2$-divergence. For any $\alpha \geq 0$ and set $X$, we pose
\begin{align*}
  & \rChi(\alpha)(X)\teq R(\dChi)(\alpha)(X) \\
  & \quad = ( \dist X, \{ (\mu, \nu)
    \mid \dChi_X(\mu, \nu) \leq \alpha, \dChi_X(\nu, \mu) \leq \alpha \})
\end{align*}
\citet[Theorem 5.4]{OlmedoThesis} shows that $\dChi$ is $H$-composable
for the monoid $H = (\R, \leq, 0, +_\chi)$, where
$\alpha +_\chi \beta = \alpha + \beta + \alpha \beta$. Hence
$\rChi(\alpha)(X)$ is a $H$-graded \Par\times assignment of $\RSRel$
on $\distmonad$ by \cref{thm:graded-comp-assign}. By applying the path
construction, we get a $H$-graded \Par\times assignment of $\Met$ on
$\distmonad$ which we can use to interpret a graded distribution type
capturing $\chi^2$-divergence:
\[
  \sem{\Tcons^{\dChi}_\alpha \tau} \teq P(\rChi(\alpha)(\sem{\tau}))
\]
The corresponding typing rule for $\bbind$ becomes:
\[
  \inferrule
  {\Gamma\vdash e_1:\Tcons^{\dChi}_{\alpha}\tau \\
    \Gamma,x:_\infty\tau\vdash e_2:\Tcons^{\dChi}_{\beta}\sigma}
  {\Gamma \vdash \bbind x\leftarrow e_1;e_2:\Tcons^{\dChi}_{\alpha +_\chi \beta}\sigma}
\]

\subsection{Further Extensions}

\paragraph*{Internalizing Group Privacy}
Differential privacy compares the results of a program when run on two input
databases at distance $1$. These guarantee can sometimes be extended to cover
pairs of inputs at distance $k$, so-called called \emph{group privacy}
guarantees. Roughly speaking, an algorithm is said to be
$(\epsilon(k), \delta(k))$-differentially private for groups of size $k$ if for
any two inputs at distance $k$, the output distributions satisfy the divergence
bound for $(\epsilon(k), \delta(k))$-differential privacy.

For standard $\epsilon$-differential privacy, group privacy is
straightforward: an $\epsilon$-private program is automatically
$k \cdot \epsilon$-private for groups of size $k$. This clean, linear
scaling of the privacy parameters is the fundamental reason why the
original Fuzz language fits $\epsilon$-differential privacy. In fact, group
privacy with linear scaling is arguably a more accurate description of the
properties captured by Fuzz---it just so happens that this seemingly stronger
property coincides with $\epsilon$-privacy.

In general, however, group privacy guarantees are not so clean. For
$(\epsilon, \delta)$-differential privacy, the parameters also degrade when
inputs are farther apart, but this degradation is not linear. In a sense, our
perspective generalizes $(\epsilon, \delta)$-differential privacy to group
privacy, a notion that better matches the linear nature of Fuzz.  For instance,
the type ${!_2} \db \multimap \Tcons_{(\epsilon, \delta)} \R$ in the last example
represents the group privacy guarantee when $(\epsilon, \delta)$-private
algorithms are applied to groups of size $2$. While we can explicitly compute
the corresponding privacy parameters, this unfolded form seems unwieldy to
accommodate in Fuzz.

\paragraph*{Handling Advanced Composition}
The typing rules we have seen so far capture two aspects of
$(\epsilon, \delta)$-differential privacy: primitives such as the Gaussian
mechanism and sequential composition via the bind rule. In practice,
$(\epsilon, \delta)$-privacy is often needed to apply the \emph{advanced
  composition theorem}~\citep{DRV10}. While standard composition simply adds up
the $(\epsilon, \delta)$ parameters, the advanced version allows trading off the
growth of $\epsilon$ with the growth of $\delta$. By picking a $\delta$ that is
slightly larger than the sum of the individual $\delta$ parameters, advanced
composition ensures a significantly slower growth in $\epsilon$.

Unlike the case of standard composition, the growth of the indices in advanced
composition is not given by a monoid operation, so it is typically applied to
blocks of $n$ programs rather than two programs at a time.  We can express this
pattern in Fuzz by adding a family of higher-order primitives
$(AC_n)_{n \in \N}$, where $AC_n$ applies advanced composition for exactly $n$
iterations.  The type of these primitives is
\[
  !_\infty(!_\infty \tau \multimap \db \multimap \Tcons_{(\epsilon, \delta)} \tau)
  \multimap (!_\infty \tau \multimap \db \multimap \Tcons_{(\epsilon^*, \delta^*)} \tau),
\]
where $\epsilon^*, \delta^*$ are as in the advanced composition theorem:
\[
  \epsilon^* = \epsilon \sqrt{2 n \ln(1/\delta')} + n \epsilon (\exp(\epsilon) - 1)
  \qquad
  \delta^* = n * \delta + \delta'
\]
for any $\delta' \in (0, 1)$.

\section{Handling Non-Termination}
\label{sec:limitations}

Most of our development would readily generalize to the full Fuzz language,
which includes general recursive types (and hence also non-terminating
expressions). In prior work \citep{AGHKC17}, we modeled the deterministic
fragment of Fuzz with metric CPOs---ordered metric spaces that support
definitions of non-expansive functions by general recursion.  We can extend this
work to encompass probabilistic features by endowing the Jones-Plotkin
probabilistic powerdomain~\citep{JP89} with metrics, much like was done in
\cref{sec:liftings}.
\aa{Say something about how we would adapt the theory of assignments
  to this in a general way.}
Briefly, the order on the probabilistic powerdomain $\mathcal{E}(X)$
is given by: $\mu \sqsubseteq \nu$ if and only if for any Scott-open
set $U$ of the CPO $X$, $\mu(U) \leq \nu(U)$ holds.  The
statistical distance and max divergence are all defined continuously
and pointwise from the probabilities $\mu(U)$, and they satisfy the
compatibility conditions required for metric CPOs. The proofs that
these distances form liftings of the probabilistic powerdomain
generalize by replacing sums over countable sets with integrals.

While the simple distances pose no major problem, the same cannot be said about
the path metric construction. A natural attempt to generalize relations to CPOs
is to require \emph{admissibility}: relations should be closed under limits of
chains to support recursive function definitions.  Unfortunately, the notion of
admissibility is not well-behaved with respect to relation composition: the
composite of two admissible relations may not be admissible. This is an obstacle
when defining the path metric, since a path of length $n$ in the graph induced
by the relation is simply a pair of points related by its $n$-fold
composition. Roughly, because admissible relations fail to compose, the path
construction does not yield metric CPOs in general, and does not form a morphism
of refinements.

The situation can be partially remedied by categorical arguments.  Both the
category of reflexive, symmetric admissible relations and the category of metric
CPOs can be characterized as fibrations over the category of CPOs with suitably
complete fibers. This allows us to define an analog of the path construction
abstractly as the left adjoint of the $Q$ functor of \cref{sec:path}, which
builds the ``at most one'' relation.  However, this construction does not
inherit the pleasant properties of the path metric on sets and functions.  More
precisely, the proof of \cref{lem:product-preservation} shown in the Appendix,
which is instrumental for showing soundness of the bind rule for
$(\epsilon, \delta)$-differential privacy, does not carry over.

\section{Related Work}
\label{sec:rw}

\paragraph*{Language-Based Techniques for Differential Privacy}
Owing to its clean composition properties, differential privacy has been a
fruitful target for formal verification. Our results build upon
Fuzz~\citep{Reed:2010}, a linear type system for differential privacy that has
subsequently been extended with sized
types~\citep{DBLP:conf/popl/GaboardiHHNP13} and algorithmic
typechecking~\citep{d2013sensitivity,AGGH14}.

Adaptive Fuzz \citep{adafuzz} is a recent extension that features an outer layer
for constructing and manipulating Fuzz programs---for instance, using program
transformations and partial evaluation---before calling the typechecker and
running the query.  By tracking privacy externally, and not in the type system,
Adaptive Fuzz supports many composition principles for
$(\epsilon, \delta)$-differential privacy, such as the advanced composition
theorem and adaptive variants called privacy filters. Our work expresses
$(\epsilon, \delta)$-privacy and basic composition in the type system, rather
than using a two-level design.

Until recently, the only type system we were aware of that could capture
$(\epsilon, \delta)$-privacy was $\mathsf{HOARe}^2$~\citep{BGGHRS15}, a
relational type system that has been extended to handle other
$f$-divergences~\citep{BFGAGHS16}.  Compared to our proposal, one drawback is
that it only provides guarantees when private algorithms are applied to inputs
that are at most a fixed distance apart.  In contrast, our sensitivity-based
approach can reason about private functions applied to inputs at arbitrary
distances.

Duet~\citep{duet}, a more recent design, proposes a two-layer type system for
handling $(\epsilon,\delta)$-privacy: one layer tracks sensitivity, analogously
to Fuzz, whereas the other layer tracks the $\epsilon$ and $\delta$ parameters
through composition.  The relationship between this system and ours is not yet
clear, but we speculate that there might be a connection between its two layers
and the path adjunction, the inner sensitivity layer corresponding to $\Met$,
and the outer layer corresponding to $\RSRel$.

Recently, variations of differential privacy have been proposed for designing
mechanisms with better accuracy.  These variations are motivated by properties
of continuous distributions and can be characterized through a span
lifting~\citep{SBGHK19}.  We hope to adapt our approach to reason about these
notions of privacy over discrete distributions, as we have for the Laplace
mechanism. An interesting problem for future work would be to extend our
semantics to continuous distributions, perhaps by leveraging recent advances in
probabilistic semantics~\citep{DBLP:journals/pacmpl/VakarKS19}.

\paragraph*{Verification of Probabilistic Relational Properties}
The last decade has seen significant developments in verification for
probabilistic relational properties other than differential privacy. Our
work is most closely related to techniques for reasoning about
$f$-divergences~\citep{BartheO13,OlmedoThesis}. Recent work by
\citet{BEGHS16} develops a program logic for reasoning about a
probabilistic notion of sensitivity based on couplings and the
Kantorovich metric. \citet{BEGHS16} identified path metrics as a
useful concept for formal verification, in connection with the path
coupling proof technique. Our work uses path metrics for a different
purpose: interpreting relational properties as function sensitivity.

The path adjunction can also be defined as a general
construction on enriched categories \citep{Lawvere1973}.  Given a
monoidal category $\mathcal{V}$ with coproducts, the forgetful functor
$\mathcal{V}-\mathsf{Cat} \to \mathcal{V}-\mathsf{Graph}$ has a left
adjoint generalizing the construction of the free category on a
graph. When $\mathcal{V} = ([0,\infty]^{\ge}, +, 0)$, a
$\mathcal{V}$-category is a metric space without the
symmetry axiom and a $\mathcal{V}$-graph is a weighted graph.
The at-most-one relation of \cref{def:at-most-one} yields a further
adjunction between $\mathcal{V}$-graphs and the category of reflexive
relations. The composition of these two adjunctions restricted to true symmetric
metric spaces and symmetric relations is precisely the path adjunction.

\paragraph*{Categorical Semantics for Metrics and Probabilities}
Our constructions build on a rich literature in categorical semantics for metric
spaces and probability theory. In prior work~\citep{AGHKC17}, we modeled the
non-probabilistic fragment of the Fuzz language using the concept of a metric
CPO; we have adapted this model of the terminating fragment of the language to
handle probabilistic sampling. \citet{DBLP:journals/entcs/Sato16} introduced a
{\em graded relational lifting} of the Giry monad for the semantics of
relational Hoare logic for the verification of $(\epsilon,\delta)$-differential
privacy with continuous distributions. Our graded liftings are similar to his
graded liftings, but the precise relationship is not yet clear.  Reasoning about
metric properties remains an active area of
research~\citep{mardare2016quantitative,plotkin-pps}.

\section{Conclusion and Future Directions}
\label{sec:conclusions}

We have extended the Fuzz programming language to handle probabilistic
relational properties beyond $\epsilon$-differential privacy, including
$(\epsilon, \delta)$-differential privacy and other properties based on
composable $f$-divergences. We introduced the categorical notion of
parameterized lifting to reason about $(\epsilon, \delta)$-differential privacy
in a compositional way. Finally, we cast relational properties as sensitivity
properties through a path metric construction.

There are several natural directions for future work. Most concretely, the
interaction between the path metric construction and non-termination remains
poorly understood. While the differential privacy literature generally does not
consider non-terminating computations, extending our results to CPOs would
complete the picture. More speculatively, it could be interesting to understand
the path construction through calculi that include adjunctions as type
constructors \citep{adjointlogic}. This perspective could help smooth
the interface between relational and metric reasoning.

\section*{Acknowledgment}

Katsumata was supported by JST ERATO HASUO Metamathematics for Systems
Design Project (No. JPMJER1603) and JSPS KAKENHI (Grant-in-Aid for
Scientific Research (C)) Grant Number JP15K00014. This work was also supported
by a Facebook TAV award.
Marco Gaboardi was partially supported by NSF under grant  \#1718220.

\bibliographystyle{plainnat} \bibliography{header,refs}

\balance

\onecolumn

\appendix

\subsection{Parameterized $L$-Relative Liftings}

Relative monads \citep{AltenkirchCU14} generalize ordinary monads by
allowing the unit and Kleisli lifting operations to depend on another
functor, akin to the role of the comonad $!$ in \eqref{eq:unitne} and
\cref{def:inftylifting}.  It is natural to wonder if parameterized
liftings are related to relative monads, but the two notions are
actually distinct. While parameterized liftings were designed to model
a bind rule under a context of variables, the Kleisli lifting of a
relative monad can only handle bind in empty contexts, which
corresponds to setting $Z$ to $\dunit$ in \cref{def:inftylifting}.

In the classical case, we can parameterize bind by combining the
Kleisli lifting of a monad with a strength.  While a notion of
strength also exists for relative monads \citep{uusutalucalco}, it was
introduced with a different purpose in mind, as an analogue of {\em
  arrows} in functional programming languages.  We add extra
conditions to relative monads so that they are above $T$, and so that
the parameterization is taken into account.
\begin{definition}\label{def:relativelifting}
  A {\em \Par\dox $L$-relative lifting of $\mT$ along
    $p:\EE\arrow\BB$} is a mapping $\Delta:|\BB|\arrow|\EE|$ such that
  1) $p\Delta X=T X$, and 2) the parameterized Kleisli lifting
  \eqref{eq:parameterized-kleisli} of $\mT$ satisfies
  \begin{displaymath}
    f:Z\dox LX\darrow \Delta Y\implies \pklift f:Z\dox\Delta X\darrow\Delta Y
  \end{displaymath}
\end{definition}
Every \Par\dox $L$-relative lifting $\Delta$ of $\mT$ satisfies
1) $\eta_{X}:LX\darrow\Delta X$ and
2) $f:LX\darrow\Delta Y$ implies $\klift f:\Delta X\darrow\Delta Y$.
From the faithfulness of $p$, these two properties imply that $\Delta$
is a $L$-relative monad.

The graded variant of \Par\dox $L$-relative lifting is defined as follows:
\begin{definition}
  An {\em $M$-graded \Par\dox $L$-relative lifting} is a monotone
  function $\Delta:(M,\le)\arrow\Ord(p,T)$ such that the parameterized
  Kleisli lifting \cref{eq:parameterized-kleisli} of $\mT$ satisfies
  \begin{displaymath}
    f:Z\dox LX\darrow \Delta \alpha Y\implies
    \pklift f:Z\dox\Delta \beta X\darrow\Delta (\beta\cdot \alpha)Y.
  \end{displaymath}
\end{definition}
This has the indistinguishability relation as an instance.
\begin{proposition}\label{pp:rdpcomp}
  Let $\R_{\ge 0}^{+}$ be the additive monoid of nonnegative real
  numbers. Then in the \Weak monoidal refinement \eqref{eq:rsrelprod},
  the indistinguishability relation, regarded as a mapping
  of type $\rDP:\R_{\ge 0}\times \R_{\ge 0}\arrow |\RSRel|$,
  satisfies
  \begin{align*}
    \rDP
    &\in\RLift\times(\distmonad,\R_{\ge 0}^{+}\times\R_{\ge 0}^{+}).
  \end{align*}
\end{proposition}
\aa{We haven't introduced this $\RLift\times$ notation for gradings before, I
  think}

% When $L$ is a strict monoidal functor, we obtain
% \begin{displaymath}
%   \theta=\pklift\eta:\Delta I\dox LJ\darrow\Delta(I\ox J),
% \end{displaymath}
% and it enjoys the axioms of {\em strong relative monads} introduced in
% \citet{uusutalucalco}.

% \jh{Don't get these last two sentences---this seems to say: relative
% monad implies relative monad?}  \sk{Is this now clear? Also, I need
% to say strong relative monads in the sense of Uustalu are
% not \Par\dox $L$-relative monads, and there is a $L$-relative monad
% which is not parameterized (I hope).}  \jh{Looks good!}

% We show that parameterized $!$-liftings and parameterized
% $L$-relative liftings are essentially same data:
% \begin{theorem}\label{thm:corr}
%   We have an equivalence of preorders:
%   \begin{displaymath}
%     \BLift\dox(\mT)\equiv \RLift\dox(\mT).
%   \end{displaymath}
% \end{theorem}

\subsection{Proof of \cref{th:eq} and \cref{th:grdeq}}

\newcommand{\Id}{\mathrm{Id}}

We show the non-graded version (\cref{th:eq}); the graded version
(\cref{th:grdeq}) is similar.

Let $\Delta\in\RLift\dox(\mT)$. We define $\dot T_\Delta X=\Delta pX$. We
show that it is a $\dox$-parameterized $!$-lifting, that is,
\begin{displaymath}
  f:X\dox LpY\darrow\Delta pZ\implies \pklift f:X\dox\Delta pY\darrow\Delta pZ.
\end{displaymath}
This is true by the assumption.

Conversely, let $\dot T\in\BLift\dox(\mT)$.  We define
$\Delta_{\dot T} I=\dot TLI$. We show that it is a $\dox$-parameterized
$L$-relative lifting, that is,
\begin{displaymath}
  f:X\dox LI\darrow\dot TLJ\implies \pklift f:X\dox\dot TLI\darrow\dot TLJ.
\end{displaymath}
This is also true from $!LI=LI$ and the assumption.

We show that the above processes are equivalence of preorders. We
first have
\begin{displaymath}
  \Delta_{\dot T_\Delta}I=\dot T_\Delta LI=\Delta pLI=\Delta I.
\end{displaymath}
Next, we have
$\dot T_{\Delta_{\dot T}}X=\Delta_{\dot T}pX=\dot TLpX\le \dot TX$.
We show the converse $\dot TX\le\dot TLpX$. Since $\dot T$ is a
$!$-lifting, we have $\eta_{pLpX}:!LpX\darrow\dot TLpX$. From
$pL=\Id$, we have $\eta_{pX}:!X\darrow\dot TLpX$. As $\dot T$ is a
$\dox$-parameterized $!$-lifting, we obtain
$\klift{\eta_{pX}}=\id_{TpX}:\dot TX\darrow\dot TLpX$, that is,
$\dot TX\le\dot TL pX$. Therefore
$\dot T\simeq\dot T_{\Delta_{\dot T}}$ holds in the preorder
$\BLift\dox(\mT)$.

This finishes the equivalence of $\BLift\dox(\mT)$ and $\RLift\dox(\mT)$.

Next, suppose that $\Delta\in\Asign\dox(\mT)$. We infer:
\begin{displaymath}
    \infer{f:X\dox FI\darrow \Delta J}{
      \infer{\lambda(f):X\darrow I\pitchfork\Delta J}{
        \infer{\lambda(f)\dox\id_{TJ}:X\dox \Delta I\darrow (I\pitchfork\Delta J)\dox \Delta I}{\pklift{ev}\circ(\lambda(f)\dox\id_{TJ}):X\dox \Delta I\darrow\Delta J}}}
\end{displaymath}
At the last step, we use the $\dox$-strong assignment. Now
\begin{align*}
  \pklift{ev}\circ(\lambda(f)\dox\id_{TJ})
  &=
    \mu\circ T(ev)\circ\theta\circ(\lambda(f)\dox T\id_J)\\
  &=
    \mu\circ T(ev\circ \lambda(f)\dox \id_J)\circ\theta\\
  &=
    \mu\circ T(f)\circ\theta\\
  &=
    \pklift f.
\end{align*}
Therefore $\Delta\in\RLift\dox(\mT)$.

Conversely, suppose that $\Delta\in\RLift\dox(\mT)$.  Since $p$ is the
map of adjunction from $-\dox LI\dashv I\dpitchfork -$ to
$-\dox I\dashv I\dmultimap -$, we have
\begin{displaymath}
  ev:(I\dpitchfork \Delta J)\dox LI\darrow\Delta J.
\end{displaymath}
Therefore we obtain
\begin{displaymath}
  \pklift {ev}:(I\dpitchfork \Delta J)\dox \Delta I\darrow\Delta J,
\end{displaymath}
that is, $\Delta\in\Asign\dox(\mT)$. We conclude that $\RLift\dox(\mT)=\Asign\dox(\mT)$.

\subsection{Proof of \cref{pp:rdpcomp}}

Before the proof, we introduce an auxiliary concept abstracting the
composition properties of divergences studied in differential privacy
as \cref{thm:dp-seq-comp}. This concept is valid when the symmetric
monoidal structure $(\dunit,\dox)$ assumed on $\EE$ in
\eqref{eq:monref} is cartesian.
\begin{definition}\label{def:compfam}
  Assume that in \eqref{eq:monref} the symmetric monoidal structure on
  $\EE$ is cartesian (hence we denote it by $(\dot 1,\dtimes)$).  An
  {\em $M$-graded sequentially composable family of $\EE$-objects}
  above $\mT$ is a monotone function
  $\Delta: {(M,\le)} \arrow\Ord(p,T)$ such that for any
  $f:Z\darrow \Delta \alpha X$ and
  $g:Z\dtimes LX\darrow \Delta \beta Y$, we have
  %
  % \jh{Is $X$ an object in the base or the total category?}
  % \sk{there was a typo - it's from base.}
  %
  \begin{displaymath}
    \pklift g\circ\langle \id_{pZ},f\rangle:Z\darrow\Delta(\alpha\cdot \beta)Y.
  \end{displaymath}
  the subpreorder of $[(M,\le),\Ord(p,T)]$ consisting of $M$-graded
  sequentially composable families of $\EE$-objects above $\mT$ is
  denoted by $\Comp(\mT,M)$.
\end{definition}
Let us see how this definition expands in the \Weak monoidal refinement
\eqref{eq:rsrelprod}, and a monad $\mT$ on
$\Set$:
\begin{displaymath}
  \xymatrix{
    (\RSRel,\dot 1,\dtimes) & & \Set \adjunction{ll}{M}{q} \ar@(ru,rd)^-\mT
  }
\end{displaymath}
An $M$-graded sequentially composabile family of $\RSRel$-objects
assigns to each set $TX$ and $\alpha,\beta\in M$, a reflexive,
symmetric relation $\sim_X^m$ satisfying:
\begin{align*}
  & (\fa{z\sim z'}f(z)\sim_X^\alpha f(z'))\\
  & \wedge(\fa{z\sim z',x}g(z,x)\sim^\beta_Y g(z',x))\\
  & \implies\fa{z\sim z'}\pklift{g}(z,f(z))\sim^{\alpha\cdot \beta}_Y\pklift{g}(z',f(z'))
\end{align*}

\begin{theorem}
  $\Asign\dtimes(\mT,M)=\Comp(\mT,M)$.
\end{theorem}
\begin{proof}
  Let $\Delta\in\Comp(\mT,M)$. That is, the following implication
  holds:
  \begin{align*}
    & f:Z\darrow\Delta\alpha X\wedge
      g:Z\dtimes LX\darrow\Delta nY\\
    & \implies\pklift g\circ\langle\id_{pZ},f\rangle:Z\darrow\Delta(m\cdot n)Y
  \end{align*}
  We instantiate the premise of the sequential composition condition
  with the following data ($X,Y$ are left unchanged)
  \begin{align*}
    Z&=(X\dpitchfork\Delta\beta Y)\dtimes\Delta\alpha X\\
    f&=\pi_1:Z\darrow\Delta\alpha X\\
    g&=\lambda(ev\circ\langle \pi_1\circ\pi_1,\pi_2\rangle):Z\dtimes LX\darrow\Delta\beta Y.
    % \lambda((f,y),x).\, f(x)
  \end{align*}
  We then obtain
  \begin{displaymath}
    \pklift g\circ\langle\id,f\rangle:
    (X\dpitchfork\Delta\beta Y)\dtimes\Delta\alpha X\darrow \Delta(\alpha\cdot\beta) Y,
  \end{displaymath}
  and we have $\kl=\pklift g\circ\langle\id,f\rangle$. Therefore
  $\Delta\in\Asign\dtimes(\mT,M)$.

  Conversely, suppose that $\Delta\in\Asign\dtimes(\mT,M)$. We have
  the following construction of morphisms in $\EE$:
  \begin{displaymath}
    \inferrule{
      \inferrule{
        \inferrule{g:Z\dtimes L X\darrow\Delta\beta Y}
        {\lambda(g):Z\darrow X\dpitchfork\Delta\beta Y} \\
        f:Z\darrow\Delta\alpha X}
      {\langle \lambda(g),f\rangle:Z\darrow (X\dpitchfork\Delta\beta Y)\dtimes\Delta\alpha X}
    }
    {\pklift g\circ\langle\id,f\rangle=
      \kl\circ\langle \lambda(g),f\rangle:Z\darrow\Delta(\alpha\cdot\beta) Y}
  \end{displaymath}
  Therefore $\Delta\in\Comp(\mT,M)$.
\end{proof}
As a result, the sequential composability (\cref{thm:dp-seq-comp}) of
differential privacy is equivalent to that the indistinguishability
relation is a graded sequentially composable family:
  \begin{align*}
    \rDP
    &\in\Comp(\distmonad,\R_{\ge 0}^{+}\times\R_{\ge 0}^{+})=
      \Asign\times(\distmonad,\R_{\ge 0}^{+}\times\R_{\ge 0}^{+})
  \end{align*}
where $\R_{\ge 0}^{+}$ is the additive monoid of nonnegative real
numbers.

\subsection{Properties of the Path Adjunction}

As a basic sanity check, converting the path metric of a relation back
into a relation with the at-most-one operation yields the original
relation.
\begin{proposition}
  \label{prop:path-inv}
  For any $X \in \RSRel$, $QPX = X$.
\end{proposition}
Composing in the opposite order $PQX$ does not yield $X$ in general.
Consider, for instance, a metric space $X=(\{x,y\},d)$ where
$d(x,y) = 2$. However, $P$ and $Q$ do form an adjoint pair.

\begin{proposition}
  \label{prop:path-adj}
  For any $X \in \RSRel$ and $Y \in \Met$, we have
  $\Met(PX, Y) = \RSRel(X, QY)$. In particular, the functors $P$ and
  $Q$ form an adjoint pair $P \dashv Q$, whose unit is identity
  morphism.
\end{proposition}

\begin{proof}
Let us show that the first term is contained in the second.  Suppose
we have a morphism $f : PX \to Y$.  Take $x$ and $x'$ in $X$ such that
$x \sim x'$.  We have $d(x, x') \leq 1$ in $PX$, thus
$d(f(x), f(x')) \leq 1$ by non-expansiveness---that is,
$f(x) \sim f(x')$ in $QY$ by definition.  This shows that $f$ is also
a morphism in $X \to QY$.

To conclude, we just need to show the reverse inclusion. Suppose we
have a morphism $f : X \to QY$.  Showing that $f$ is a non-expansive
function $PX \to Y$ is equivalent to showing that, given a path
$x_0 \sim \cdots \sim x_k$ of length $k$ in $X$, $d(x_0, x_k) \leq k$.
Since $f : X \to QY$, we know that for every $i < k$ the relation
$f(x_i) \sim_{QY} f(x_{i+1})$ holds; by definition, this means that
$d(f(x_i), f(x_{i+1})) \leq 1$.  Applying the triangle inequality
$k - 1$ times, we conclude $d(f(x_0), f(x_k)) \leq k$.
\end{proof}

Since the unit is identity, we conclude that $P$ is {\em full and
  faithful}.  In the current setting, this fact can be phrased as
follows:
\begin{corollary} \label{cor:path-faithful} Let $X$ and $Y$ be objects
  of $\RSRel$. Then $f:X\rearrow Y$ if and only if $f:PX\nearrow PY$.
\end{corollary}

% \begin{proof}
%   Immediate. $\Met(PX,PY)\simeq\RSRel(X,QPY)=\RSRel(X,Y)$.
% \end{proof}

The adjunction $P \dashv Q$ preserves much---but not all---of the structure in
$\RSRel$ and $\Met$. Combined with the previous results, these properties yield
a proof of \cref{th:transpath}. We summarize these properties below.  First,
discrete spaces are preserved.

\begin{lemma}
  \label{lem:discreteness-preservation}
  For any set $X$, $P(\infty\cdot X) = \infty\cdot X \in \Met$, and
  $Q(\infty\cdot X) = \infty\cdot X \in \RSRel$.
\end{lemma}

The path functor $P:\RSRel\arrow\Met$ strictly preserves all products
(including infinite ones), and symmetric monoidal structure.
% \sk{This needs to be small, otherwise we fail to preserve powers.}
\begin{lemma}
  \label{lem:product-preservation}
  Let $(X_i)_{i \in I}$ be a family of objects of $\RSRel$.  Then
  $P\left(\prod_i X_i\right) = \prod_i PX_i$. The metric on a
  general product of metric spaces is defined by taking the supremum
  of all the metrics.
\end{lemma}

\begin{proof}
As the underlying sets are equal, it suffices to show that the metrics
are equal.  Let $d_P$ be the metric associated with
$P\left(\prod_i X_i\right)$, and let $d_{\prod}$ be the metric
associated with $\prod_i PX_i$.  Since $P$ is a functor, we know that
the identity function is a morphism in
$P\left(\prod_i X_i\right) \to \prod_i PX_i$; that is,
$d_{\prod}(x, x') \leq d_P(x, x')$ for all families
$x, x' \in \prod_i X_i$.  To conclude, we just need to show the
opposite inequality. If $d_{\prod}(x, x') = \infty$, we are done;
otherwise, there exists some $k \in \N$ such that
$d_{PX_i}(x_i, x'_i) \leq k$ for every $i \in I$.  Since all the
relations associated with the $X_i$ are reflexive, we know that for
every $i$ there must exist a path of related elements
$x_i = x_{i,0} \sim \cdots \sim x_{i,k} = x'_i$ of length exactly $k$,
by padding one of the ends with reflexivity edges.  Therefore, there
exists a path of related families
$(x_i)_{i \in I} = (x_{i,0})_{i \in I} \sim \cdots \sim (x_{i,k})_{i
  \in I} = (x'_i)_{i \in I}$ of length $k$ in $\prod_i X_i$.  By the
definition of the path metric, this means that $d_P(x, x') \leq k$,
and thus $d_P(x, x') \leq d_{\prod}(x, x')$, as we wanted to show.
\end{proof}

% \begin{corollary}
%   \label{lem:function-preservation}
%   $P(X\pitchfork Y) = X\pitchfork PY$.
% \end{corollary}

\begin{lemma}
  \label{lem:tensor-preservation}
  For every $X, Y \in \RSRel$, $P(X \otimes Y) = PX \otimes PY$.
\end{lemma}

\begin{proof}
Once again, it suffices to show that both metrics are equal.

For all pairs $(x, y)$ and $(x', y')$ in $X \times Y$, there are paths
between $x$ and $x'$ and between $y$ and $y'$ whose summed length is
at most $k$ if and only if there is a path between $(x, y)$ and
$(x', y')$ of length at most $k$ in $X \otimes Y$.  For the ``only
if'' direction, if there are such paths $(x_i)_{i \leq k_x}$ and
$(y_i)_{i \leq k_y}$, then the sequence
\[ (x_0, y_0), (x_1, y_0), \ldots, (x_{k_x}, y_0), (x_{k_x}, y_1),
  \ldots, (x_{k_x}, y_{k_y}) \] is a path from $(x, y)$ and $(x', y')$
in $X \otimes Y$. Conversely, suppose that we have a path from
$(x, y)$ to $(x', y')$ of length at most $k$ in $X \otimes Y$.  We can
show by induction on the length of the path that there are paths from
$x$ to $x'$ and from $y$ and $y'$ with total length at most $k$.  The
base case is trivial.  If there is a hop, by the definition of the
relation for $X \otimes Y$, this hop only adds one unit to the length
of the path on $X$ or to the path on $Y$, which allows us to conclude.
\end{proof}

\subsection{Proof of \cref{th:mapasign}}

Let $\Delta:(M,\le)\arrow\Ord(p,T)$ be an $M$-graded \Par\dox
assignment of $\EE$ on $\mT$.  Since $p\circ F=p'$, $F\circ\Delta$ is
a monotone function of type $(M,\le)\arrow\Ord(p',T)$. By applying $F$
to the internal Kleisli lifting morphism \cref{eq:internal-kleisli},
we obtain
\begin{displaymath}
  \inferrule{kl:(X\pitchfork \Delta m Y)\dox \Delta n X\darrow \Delta(n\cdot m) Y}{
    kl:(X\mathbin{\ddot\pitchfork} F(\Delta m Y))\ddox F(\Delta n X)\darrow F(\Delta(n\cdot m) Y)}
\end{displaymath}
This concludes that $F\circ \Delta:(M,\le)\arrow\Ord(p',T)$ is
a \Par\ddox assignment.

\subsection{Modeling Other Divergences in Graded Fuzz}
\label{app:other-divegences}

\paragraph*{Hellinger Distance}
The \emph{Hellinger} distance is a standard measure of similarity between
distributions originating from statistics. For discrete distributions over $X$,
the Hellinger distance is defined as:
\[
  \dHellinger_X(\mu, \nu) \teq \sqrt{ \frac{1}{2} \sum_{i \in X}
    |\sqrt{\mu(i)} - \sqrt{\nu(i)}|^2 } .
\]
This distance is a proper metric: it satisfies reflexivity, symmetry,
and triangle inequality. However, its behavior under composition means
that we cannot model it with a standard assignment structure. Instead,
we will use a graded assignment. For any $\alpha \geq 0$ and set $X$,
we define:
\begin{align*}
  \rHD(\alpha)(X) & \teq R(\dHellinger)(\alpha)(X) \\
                  & = ( \dist X, \{ (\mu,
                    \nu) \mid \dHellinger_X(\mu, \nu) \leq \alpha, \dHellinger_X(\nu,
                    \mu) \leq \alpha \} )
\end{align*}
Note that this is a reflexive and symmetric relation, hence an object
in $\RSRel$. \citet[Proposition 5]{BartheO13} show the following
composition principle:
\[
  \dHellinger_Y(f^\dagger \mu, g^\dagger \nu)^2 \leq
  \dHellinger_X(\mu, \nu)^2 + \sup_{x \in X} \dHellinger_Y(f(x),
  g(x))^2
\]
Note that this is for \emph{squared} Hellinger distance. Taking square
roots, we have the following composition property for standard
Hellinger distance:
\[
  \dHellinger_Y(f^\dagger \mu, g^\dagger \nu) \leq
  \sqrt{\dHellinger_X(\mu, \nu)^2 + \sup_{x \in X} \dHellinger_Y(f(x),
    g(x))^2}
\]
Hence, $\dHellinger$ is $H$-composable for the monoid
$H = (\R, \leq, 0, +_2)$, where
$\alpha +_2 \beta = \sqrt{ \alpha^2 + \beta^2 }$ and
$\rHD(\alpha)(X)$ has the structure of a $H$-graded
\Par\times
assignment of $\RSRel$ on $\distmonad$ by
\cref{thm:graded-comp-assign}.  By applying the path construction to
this assignment, we get a $H$-graded \Par\times assignment of $\Met$
on $\distmonad$ with which we can then interpret a graded
distribution type capturing Hellinger distance:
\[
  \sem{\Tcons^{\dHellinger}_\alpha \tau} \teq
  P(\rHD(\alpha)(|\sem{\tau}|))
\]
The typing rule for $\bbind$ is then:
\[
  \inferrule
  {\Gamma\vdash e_1:\Tcons^{\dHellinger}_{\alpha}\tau \\
    \Gamma,x:_\infty\tau\vdash e_2:\Tcons^{\dHellinger}_{\beta}\sigma}
  {\Gamma \vdash \bbind x\leftarrow
    e_1;e_2:\Tcons^{\dHellinger}_{\alpha +_2 \beta}\sigma}
\]
Like the other divergences, we can introduce primitives and typing
rules capturing the Hellinger distance. For example, the
natural-valued \emph{Poisson} distribution models the
number of events occurring in some time interval, if the events happen
independently and at constant rate. Given a parameter $\alpha : \N$,
this distribution has the following probability mass function:
\[
  P(\alpha)(n) = \frac{\alpha^n e^{-\alpha}}{n!}
\]
It is known that the Poisson distribution satisfies the following
Hellinger distance bound:
\[
  \dHellinger_{\N}(P(\alpha), P(\alpha')) = 1 - \exp \left(-
    \frac{|\sqrt{\alpha} - \sqrt{\alpha'}|^2}{2}\right) \leq
  \frac{1}{2} |\sqrt{\alpha} - \sqrt{\alpha'}|^2
\]
Given $\alpha$ and $\alpha'$ at most $1$ apart, this Hellinger
distance is at most $1/2$.  Hence, we may introduce a type $\poiD$ and
interpret it as $\sem{\poiD} = \lambda x.\, P(x)$, and the following
typing rule is sound:
\[
  \inferrule { } {\Gamma \vdash \poiD : \lceil \R \rceil \multimap
    \Tcons^{\dHellinger}_{1/2} \N}
\]

\end{document}